\documentclass[reprint, amsfonts, amssymb, amsmath,  showkeys, nofootinbib,pra, superscriptaddress, twocolumn,longbibliography,floatfix]{revtex4-2}
\usepackage{float}
\makeatletter
\let\newfloat\newfloat@ltx
\makeatother
\usepackage[english]{babel}

\makeatletter 
\renewcommand\onecolumngrid{
\do@columngrid{one}{\@ne}
\def\set@footnotewidth{\onecolumngrid}
\def\footnoterule{\kern-6pt\hrule width 1.5in\kern6pt}
}
\renewcommand\twocolumngrid{
        \def\footnoterule{
        \dimen@\skip\footins\divide\dimen@\thr@@
        \kern-\dimen@\hrule width.5in\kern\dimen@}
        \do@columngrid{mlt}{\tw@}
}
\makeatother

\usepackage[utf8]{inputenc}
\usepackage{graphics}
\usepackage{selinput}
\usepackage{bbm}
\usepackage{braket}
\usepackage{amsthm}
\usepackage{mathtools}
\usepackage{bm}
\usepackage{physics}
\usepackage[dvipsnames]{xcolor}
\usepackage{graphicx}
\usepackage[left=16mm,right=16mm,top=35mm,columnsep=15pt]{geometry} 
\usepackage{adjustbox}
\usepackage{placeins}
\usepackage[T1]{fontenc}
\usepackage{lipsum}
\usepackage{csquotes}
\usepackage{mathbbol} 
\usepackage[linesnumbered,ruled,vlined]{algorithm2e}
\usepackage[tight,footnotesize]{subfigure}

\SetKwInput{kwInit}{Init}

\def\HC{\mathcal{H}}

\def\LC{\mathcal{L}}

\def\ad{^{\dagger}}

\def\w{\omega}

\usepackage[makeroom]{cancel}
\usepackage[toc,page]{appendix}
\usepackage[colorlinks=true,citecolor=blue,linkcolor=magenta]{hyperref}

\usepackage{tikz}
\tikzset{every picture/.style=remember picture}

\newcommand{\poly}{\operatorname{poly}}

\newcommand{\Ebb}{\mathbb{E}}

\newcommand{\Ubb}{\mathbb{U}}

\newcommand{\AC}{\mathcal{A}}

\newcommand{\EC}{\mathcal{E}}

\newcommand{\IC}{\mathcal{I}}
\newcommand{\MC}{\mathcal{M}}
\newcommand{\NC}{\mathcal{N}}
\newcommand{\OC}{\mathcal{O}}
\newcommand{\PC}{\mathcal{P}}
\newcommand{\QC}{\mathcal{Q}}

\newcommand{\SC}{\mathcal{S}}

\newcommand{\UC}{\mathcal{U}}
\newcommand{\VC}{\mathcal{V}}

\newcommand{\XC}{\mathcal{X}}

\newcommand{\Var}{{\rm Var}}
\newcommand{\Cov}{{\rm Cov}}

\renewcommand{\geq}{\geqslant}
\renewcommand{\leq}{\leqslant}

\newcommand{\rhot}{\widetilde{\rho}}

\renewcommand{\vec}[1]{\boldsymbol{#1}}  

\newcommand{\bs}{\textsf{BS}}

\newcommand{\thv}{\vec{\theta}}

\newcommand{\xv}{\vec{x}}

\newtheorem{innercustomthm}{Theorem}
\newenvironment{customthm}[1]
  {\renewcommand\theinnercustomthm{#1}\innercustomthm}
  {\endinnercustomthm}

\def\be{\begin{equation}}
\def\ee{\end{equation}}

\def\bs{\begin{split}}
\def\es{\end{split}}
\def\bea{\begin{eqnarray}}
\def\eea{\end{eqnarray}}
\def\w{\omega}

\newcommand\X{\text{X}}

\newtheorem{theorem}{Theorem}
\newtheorem{lemma}{Lemma}
\newtheorem{corollary}{Corollary}

\newtheorem{proposition}{Proposition}
\newtheorem*{proposition*}{Proposition}

\newtheorem{definition}{Definition}

\newenvironment{specialproof}{\textit{Proof:}}{\hfill$\square$}

\usepackage{amssymb}
\usepackage{dsfont}

\usepackage[normalem]{ulem}
\usepackage{hyperref}

\usepackage[defaultcolor=red]{changes}

\begin{document}

\title{On fundamental aspects of quantum extreme learning machines}

\author{Weijie Xiong}
\thanks{The first three authors contributed equally to this work.}
\affiliation{Institute of Physics, Ecole Polytechnique F\'{e}d\'{e}rale de Lausanne (EPFL), Switzerland}

\author{Giorgio Facelli}
\thanks{The first three authors contributed equally to this work.}
\affiliation{Institute of Physics, Ecole Polytechnique F\'{e}d\'{e}rale de Lausanne (EPFL), Switzerland}

\author{Mehrad Sahebi}
\thanks{The first three authors contributed equally to this work.}
\affiliation{Institute of Electrical and Micro Engineering, Ecole Polytechnique F\'{e}d\'{e}rale de Lausanne (EPFL), Switzerland}

\author{Owen Agnel}
\affiliation{Department of Computer Science, University of Oxford, Oxford, UK}

\author{Thiparat Chotibut}
\affiliation{Chula Intelligent and Complex Systems Lab, Department of Physics, Faculty of Science, Chulalongkorn University, Bangkok, Thailand}

\author{Supanut Thanasilp}
\affiliation{Institute of Physics, Ecole Polytechnique F\'{e}d\'{e}rale de Lausanne (EPFL), Switzerland}

\author{Zo\"{e} Holmes}
\affiliation{Institute of Physics, Ecole Polytechnique F\'{e}d\'{e}rale de Lausanne (EPFL), Switzerland}

\date{\today}

\begin{abstract}
Quantum Extreme Learning Machines (QELMs) have emerged as a promising framework for quantum machine learning. Their appeal lies in the rich feature map induced by the dynamics of a quantum substrate – the quantum reservoir – and the efficient post-measurement training via linear regression. Here we study the expressivity of QELMs by decomposing the prediction of QELMs into a Fourier series. We show that the achievable Fourier frequencies are determined by the data encoding scheme, while Fourier coefficients depend on both the reservoir and the measurement. Notably, the expressivity of QELMs is fundamentally limited by the number of Fourier frequencies and the number of observables, while the complexity of the prediction hinges on the reservoir. As a cautionary note on scalability, we identify four sources that can lead to the exponential concentration of the observables as the system size grows (randomness, hardware noise, entanglement, and global measurements) and show how this can turn QELMs into useless input-agnostic oracles. {In particular, our result on the reservoir-induced concentration strongly indicates that quantum reservoirs drawn from a highly random ensemble make QELM models unscalable.} Our analysis elucidates the potential and fundamental limitations of QELMs, and lays the groundwork for systematically exploring quantum reservoir systems for other machine learning tasks.
\end{abstract}

\maketitle

\section{Introduction}
Extreme learning machines (ELMs)~\cite{huang2004extreme, Ding2015extreme, wang2022review, HUANG2015ELMRev} are a class of feed-forward neural networks designed to address the challenges of conventional neural networks training. While gradient-based training methods such as backpropagation have propelled the rapid development of neural network models, their efficiency can diminish with increasing data and model size due to issues like vanishing gradients and  increasing computational resources ~\cite{goodfellow2016deep}. ELMs leverage a single hidden layer with numerous randomly initialized hidden neurons, whose parameters are subsequently fixed to create a rich representation of the input data. Training ELMs reduces to a one-step linear regression on the output layer weights, which is simply a convex optimization problem. This approach can significantly reduce training time and even improve generalization performance~\cite{huang2006extreme,huang2011extreme, HUANG2015ELMRev}.

\begin{figure}
    \centering
    
    \includegraphics[width=0.48\textwidth]{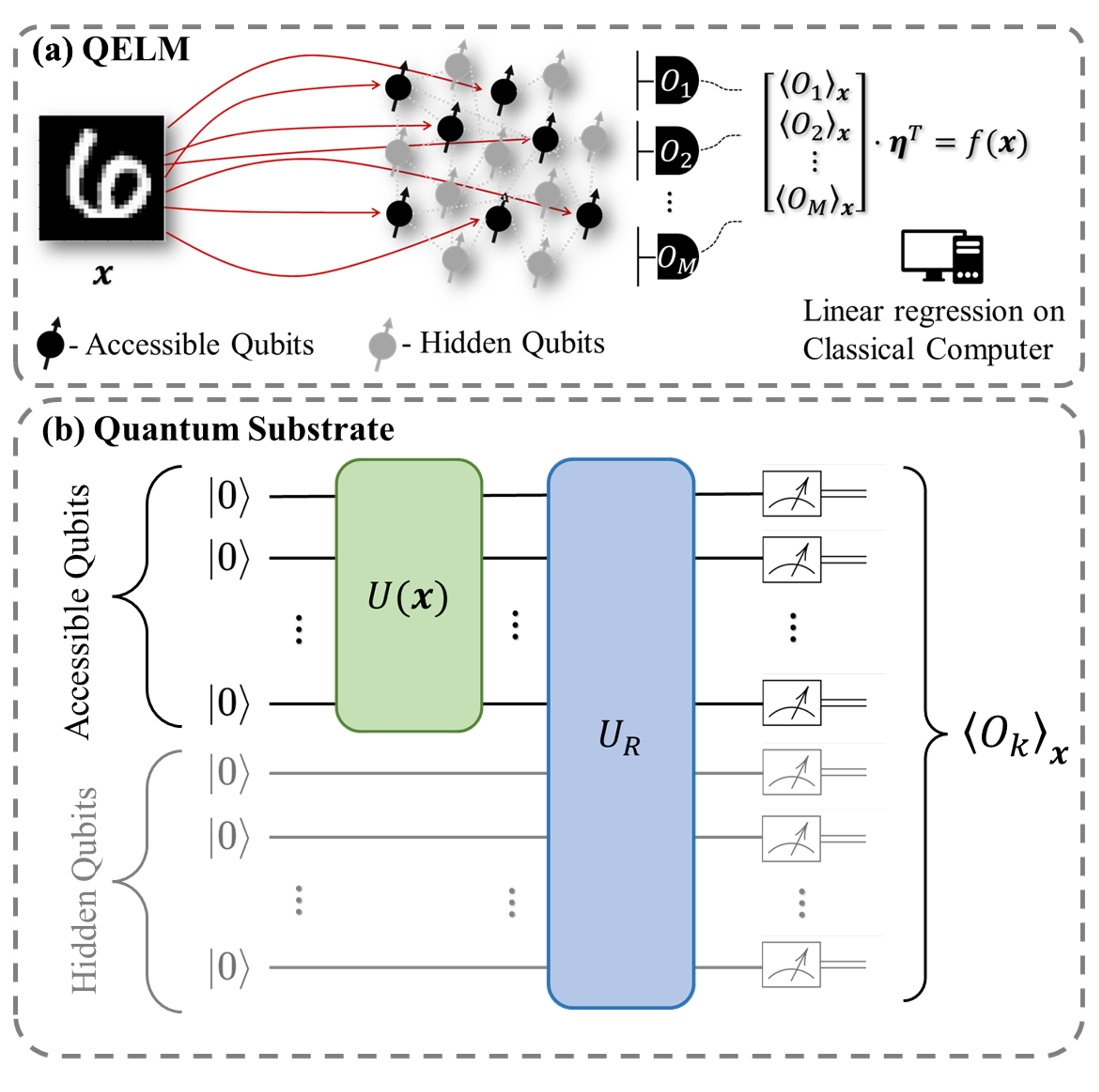}
    
    \caption{\textbf{Framework of a QELM.} A QELM encodes classical data onto accessible qubits (red arrows in \textbf{panel (a)} or the encoding unitary $U(\xv)$ in \textbf{panel (b)}). The accessible and hidden qubits then undergo a unitary (reservoir) evolution $U_R$, and a set of {Hermitian} observables {$\{O_k\}_{k=1}^{M}$} is measured. The estimates of the observables are classically post-processed to predict the output {$f(\xv)$} via linear regression. The reservoir unitary is fixed and only the linear regression weights $\bm{\eta}$ are classically optimized.}
    \label{fig:Framework}
\end{figure}

With the advent of near-term quantum devices, Quantum Extreme Learning Machines (QELMs) have emerged as a compelling alternative to traditional ELMs~\cite{mujal2021opportunities, innocenti2023potential}. This growing interest stems not only from the capability of QELMs to directly process quantum data~\cite{innocenti2023potential, Ghosh_PRL_CQQ_2019, Ghosh_NPJQI_QQC_2019,suprano2023experimental}, but also from the fact that QELMs can leverage the complex dynamics of a quantum system evolving in an exponentially large Hilbert space -- a quantum reservoir -- to construct an intricate feature map of classical data inputs. Employing a quantum reservoir enables QELMs and more generally quantum reservoir computing to achieve predictive power with significantly lower resource requirements, compared to extensive hidden layers required in the classical counterparts ~\cite{mujal2021opportunities, fujii2017harnessing, nakajima_PRApp_spatialmultiplexing_2019}. Notably, despite involving measurement outcomes from the quantum reservoir, training QELMs remains a convex optimization problem based on linear regression. This guarantees trainability even on noisy quantum hardware, unlike gradient-based {Variational Quantum Algorithms (VQAs)} which can suffer from noise-induced barren plateaus~\cite{wang2020noise}. These advantages suggest that QELMs hold promise for harnessing the potential of near-term quantum devices in various quantum machine learning tasks. 

However, for any Quantum Machine Learning (QML) model learning on classical data, it is important to consider the expressivity limitations arising from the data encoding strategy~\cite{schuld2021effect}. In many cases, the predictions of a typical QML model can be expressed in terms of a Fourier series~\cite{schuld2021effect}. Hence, the quantity of achievable Fourier frequencies and the controllability of {the trainable weights over} the Fourier coefficients become natural measures of a QML model's expressivity and predictive power. As a rule of thumb, the less expressive a QML model is, the more efficiently it can be classically simulated. Hence, the Fourier decomposition provides a way to assess its classical simulability, in particular whether it is possible to construct a classical surrogate of the quantum model, and hence the lack of (or potential for) quantum advantage~\cite{schreiber2022classical,landman2022classically,sweke2023potential}.

In parallel, there is growing awareness of the dangers posed by exponential concentration, where quantum expectation values (e.g. of a loss or a kernel) exponentially concentrate around an input-independent value with growing problem size~\cite{mcclean2018barren,holmes2021connecting,cerezo2020cost,sharma2020trainability,wang2020noise,cerezo2020impact,holmes2020barren,marrero2020entanglement,uvarov2020barren,arrasmith2020effect,arrasmith2021equivalence,thanasilp2022exponential, rudolph2023trainability, ragone2023unified,fontana2023theadjoint,garcia2023deep, cerezo2023does}. The main sources of exponential concentration are ansatz Haar-expressivity~\cite{mcclean2018barren}, noise~\cite{wang2020noise}, entanglement~\cite{marrero2020entanglement} and global measurements~\cite{cerezo2020cost}. Originally identified as a problem for quantum neural networks (QNNs)~\cite{mcclean2018barren}, it has recently been shown that exponential concentration, analogously, acts as a barrier to the scalability of quantum kernel methods~\cite{thanasilp2022exponential, gan2023unified}. While for QNNs exponential concentration inhibits the training of the model, for quantum kernel methods the trainability is guaranteed but it is the generalization capabilities of the model that suffer. 

This work aims to explore the potential and limitations of QELM through the lenses of Fourier-expressivity, exponential concentration, and classical simulability. We focus on the framework of a general QELM that learns from classical data and performs regression or classification tasks from the measurement outcomes of a quantum reservoir. Under this framework, we first show that the prediction of a QELM can be expressed as a Fourier series. Similar to a QNN~\cite{schuld2021effect}, the achievable Fourier frequencies of a QELM are also fully determined by the encoding scheme, i.e. the eigenvalues of the Hermitian generators of encoding unitary. The Fourier coefficients are input-independent and controlled by both the dynamics of the QELM and the measurement. Based on that, we investigate the classical simulability for typical encoding strategies. In particular, the encoding using Pauli generators leads to a prediction composed of linearly many Fourier modes in the system size, which can be efficiently surrogated by a classical computer~\cite{schreiber2022classical,landman2022classically,sweke2023potential}. Hence, the use of an exponential encoding, which results in exponentially many Fourier frequencies, is a prerequisite for quantum advantage. We further link the Fourier-expressivity of QELM to the independence of the Fourier coefficients, which is upper bounded by three factors: the number of Fourier frequencies determined by the encoding strategy, the number of observables, and the measurement locality.

Then, we further study the effect of exponential concentration on QELM. In particular, QELM implements a set of observables, whose outcomes are then used to train the final linear regression. We show that the Haar-expressivity of both the encoding and reservoir unitaries, entanglement, global measurement, and noise can independently induce the exponential concentration of the observables' expectation values about an input-independent value. This, in turn, necessitates exponentially many shots to precisely estimate the observable's outcomes and reliably reconstruct the  input-dependent prediction. Crucially, these four sources of concentration limit the scalability of a QELM, since the model's predictions will become increasingly agnostic to the input data as the system size grows. {In particular, our analytical results strongly discourage the use of the reservoirs drawn from 2-designs and their approximations~\cite{vetrano2024state,domingo2023optimal,ahmed2024optimal,rodrigo2021dynamicalptqrc,fujii2017harnessing}, which lead to the concentration and in turn the unscalability of QELMs.} Our key contributions and related analytical work are tabulated in Table \ref{table:KernelsVsQML}.

\begin{table*}
\resizebox{\width}{!}{%
\begin{tabular}{c|c|c|c|}
\cline{2-4}
\multicolumn{1}{l|}{} & Expressivity & Controllability & Exponential concentration \\ \hline
\multicolumn{1}{|c|}{Classical data} & \begin{tabular}[c]{@{}c@{}} Sec.~\ref{fourier_decomposition_analysis}: {Eq.~\eqref{Claim1Eq}} \\ Ref. {\cite{goto2021universal,gonon2023universal}}\end{tabular} & Sec.~\ref{fourier_decomposition_analysis}: {Eqs.~\eqref{Eq:Control1}, \eqref{Eq:Control2}} & \begin{tabular}[c]{@{}c@{}}Sec.~\ref{trainability_analysis}:\\Encoding - Eqs.~\eqref{Eq:Conc_EncodingExpress}, \eqref{eq:noisy_concentration}\\Reservoir - Eqs.~\eqref{Eq:Conc_Res_Express}, \eqref{Eq:Conc_Ent}\\ Measurement - Eq.~\eqref{Eq:Conc_Gl_Meas}\end{tabular} \\ \hline
\multicolumn{1}{|c|}{Quantum data} & Ref. {\cite{innocenti2023potential}} & Ref. {\cite{innocenti2023potential}} & \begin{tabular}[c]{@{}c@{}}Sec.~\ref{trainability_analysis}: \\Reservoir - Eqs.~\eqref{Eq:Conc_Res_Express}, \eqref{Eq:Conc_Ent}\\ Measurement - Eq.~\eqref{Eq:Conc_Gl_Meas}\end{tabular} \\ \hline
\end{tabular}%
}
{\caption{\textbf{Summary of our main analytical results and related analytical work:}
The controllability generally refers to the scope of QELM's predictive capabilities that can be adjusted by controlling the trainable weights. For classical data, we show that the output prediction is a linear combination of Fourier series; hence, the controllability of trainable weights here means the controllability {over} Fourier coefficients.
}}
\label{table:KernelsVsQML}
\end{table*}

\section{Background}

Quantum extreme learning machines (QELMs) can be used to solve both regression and classification tasks. Their appeal can be attributed to the ability of QELMs to encode input data in exponentially large many-body quantum states induced by the dynamics of a quantum reservoir. This high-dimensional input feature map facilitates efficient classification and linear regression, similar to the concept underlying kernel methods. 

Recent studies have demonstrated the potential of QELMs across diverse quantum machine learning tasks, encompassing binary classification~\cite{fujii2021quantum}, supervised learning on benchmark datasets~\cite{wang2022variational,sainz2022quantum,sakurai2022quantum}, input recognition and parity check~\cite{negoro2021toward}, and quantum state reconstruction~\cite{innocenti2023potential,suprano2023experimental}.  Notably, QELMs can process both  quantum~\cite{innocenti2023potential,suprano2023experimental,Ghosh_NPJQI_QQC_2019,Ghosh_PRL_CQQ_2019} and classical data~\cite{goto2021universal,gonon2023universal,fujii2021quantum,wang2022variational,sainz2022quantum,negoro2021toward,sakurai2022quantum}. However, theoretical investigation on their potential and limitations when processing classical data remains relatively unexplored. This motivates our current study concerning the consequences of the interplay between classical data encoding and exponential concentration phenomena in QELMs. 

We consider the following framework of QELMs. Given a training set $\{ (\xv^{(l)}, \vec{y}^{(l)}) \}_{l=1}^{D} \subset \mathbb{R}^{d_x} \times \mathbb{R}^{d_y}$ consisting of $D$ input-output classical vector pairs, the objective of a QELM is to learn a function $\vec{f}$ such that $\vec{f}(\xv) \approx \vec{y} $. For brevity, we focus our analysis on a scalar output ($d_y = 1$), as the generalization to the vector output case is more cumbersome and existing proposals are mostly based on a scalar output case. 

As mentioned earlier, at its core, a QELM harnesses a quantum reservoir, e.g. quantum spins~\cite{innocenti2023potential}, crystal model~\cite{sakurai2022quantum}, photonic~\cite{suprano2023experimental,nerenberg2024photon} and NMR~\cite{negoro2021toward} platform, as a means of implementing a rich feature map to the input data. These reservoir states that encode the inputs are then read out via measurements, yielding classical vectors of observables that are subsequently used to train the model classically via linear regression to match the outputs (labels). More concretely, as shown in Fig.~\ref{fig:Framework}, a QELM consists of the following components:
\begin{enumerate}
    \item \textbf{Encoding of classical data:} An input vector $\xv \in \XC =\{\xv^{(l)}\}_{l=1}^{D}$ is encoded into a quantum state via a parametrized unitary $U(\xv)$ applied on the space of $n_{a}$ accessible qubits. Let $\rho_0$ denote the initial state of all accessible qubits, then the state after encoding is
    \begin{equation}
        \rho(\xv)=U(\xv)\rho_{0} U^\dag(\xv) \,.
    \end{equation}
    
    \item \textbf{Reservoir evolution:} As well as the accessible qubits, a QELM is equipped with a reservoir composed of $n_{h}$ hidden qubits. We suppose the hidden qubits are initialized in the $\ket{0}$ state. The second step of a QELM is to apply the reservoir dynamics to the composite system of accessible and hidden qubits. Without loss of generality, we consider the reservoir evolution described by some unitary $U_R$. The state of the reservoir after the evolution is
    \begin{equation}\label{Eq:ReservoirEvo}
        \rhot(\xv)=U_R\big(\rho(\xv)\otimes|0\rangle\langle 0|\big)U_R^\dag \,.
    \end{equation}
    
    Indeed, by virtue of Stinespring's dilation theorem, supposing that the reservoir undergoes a {Completely-Positive Trace Preserving (CPTP)} channel evolution would only correspond to some unitary evolution on a larger hidden space. Thus, we will restrict ourselves only to unitary dynamics.
    
    \item \textbf{Readout:} Next, the reservoir state after the evolution is read out by  measurements on a set of observables $\{O_{1}, O_{2}, \dots, O_{M}\}$, whose theoretical expectation value is
    \begin{equation}\label{obs_def}
        \expval{O_{k}}_{\xv}=\operatorname{Tr}[O_{k} \rhot(\xv)] \,.
    \end{equation}
    These theoretical expectation values yield the measurement readout (classical) vector. 

    \item \textbf{Linear regression:} Finally, the resulting measurement readout vectors are classically trained via linear regression to match the corresponding prediction labels. That is, given a set of trainable weights $\bm{\eta}=[\eta_1, \eta_2, ... ,\eta_M]^T$, a QELM makes the following prediction
    \begin{equation}\label{Eq:sysout}
        f_{\bm{\eta}}(\xv)=\sum_{k=1}^M \eta_k \expval{O_{k}}_{\xv} \,.
    \end{equation}
    The weight vector $\bm{\eta}$ is typically trained by minimizing the empirical loss 
    \begin{equation}
        \mathcal{L}(\bm{\eta}) = \frac{1}{D} \sum_{k=1}^{D}  \left( f_{\bm{\eta}}\big(\xv^{(l)}\big) - y^{(l)} \right)^2 \,.
    \end{equation}
\end{enumerate}

\section{Fourier decomposition analysis}\label{fourier_decomposition_analysis}

\subsection{General analysis}
To analyse the Fourier-expressivity of a quantum model it is helpful to consider its Fourier decomposition. For simplicity we will here assume that the inputs are scalars, i.e., $\vec{x} \rightarrow x$, but in Appendix~\ref{AppMultiFourier} we discuss the generalization to vector inputs.

Ref.~\cite{schuld2021effect} showed that the output of a general variational quantum circuit that encodes classical data via parameterized unitaries can be expressed as a Fourier series of the form 
\begin{equation}
    f_{\bm{\theta}} (x)=\sum_{\omega\in\Omega}c_{\omega} e^{i\omega x} \,. 
\end{equation}
The set of frequencies $\Omega$ is determined by the encoding strategy and the coefficients $c_\omega$ are determined by the corresponding {Variational Quantum Circuit (VQC)} parameters. Here we show that the prediction of a QELM can also be expressed as a Fourier series.

Let us suppose that the accessible qubits are initialized in the state \\ 
\begin{equation}
    \rho_{0}= \sum_{i,j} \alpha_{ij}|i\rangle\langle j|\,.
\end{equation}
Here we use $\ket{i}$ to denote the state of the effective qudit corresponding to the accessible qubits, i.e. we denote the basis $\{\ket{00\cdots00},\ket{00\cdots01},\dots, \ket{11\cdots11}\}$ as $\{\ket{0},\ket{1}, \dots,  \ket{2^{n_{a}}-1}\}$. {We consider here the ``time-evolution'' encoding, i.e.} the encoding unitary has the form
\begin{equation}\label{eq:encoding_def}
    U(x)= e^{iHx}\,,
\end{equation}
where the generator $H$ is a Hermitian observable. {This type of encoding strategy has been widely used in the literature of quantum machine learning with classical data, as discussed in ref.~\cite{schuld2021effect} and the references therein}. {Particularly, in the context of QELMs and QRC model~\cite{fujii2017harnessing,mujal2021opportunities, wang2022variational,varsamopoulos2024quantum,de2024harnessing}, the commonly used amplitude encoding on a single qubit can be analyzed under the framework of ``time-evolution'' encoding with classical pre-processing.}

{This class of encoding strategies not only encompasses the cases that the classical inputs are directly encoded in evolution time of a general Hamiltonian, but also can be easily modified to those implicitly containing a step of classical pre-processing of data $\xv\mapsto\phi_{\rm pre}(\xv)$ where $\phi_{\rm pre}(\cdot)$ is a pre-processing map. In particular, we have the Fourier series with respect to $\phi_{\rm pre}(\xv)$ (instead of $\xv$). This can lead to a significant increase in frequencies of the Fourier spectrum as in the exponential encoding we discuss in Sec.~\ref{subsec:encoding}, or change the basis form of the model prediction as in the case of the amplitude encoding using single qubit. For a detailed discussion, we refer the readers to Sec.~V\,A of Ref.~\cite{schuld2021effect}.}

We further assume an appropriate computational basis, such that the generator $H$ is a diagonal Hamiltonian with eigenvalues $\{\lambda_0, \lambda_1, \dots, \lambda_{(2^{n_a}-1)}\}$. It follows that the state of accessible qubits after the encoding is
\begin{align} \label{eq:10}
     \rho(x) &= U(x) \rho_0 U(x)^\dag\\
     &= \sum_{i,j=0}^{{2^{n_a}}-1} e^{i(\lambda_{j}-\lambda_{i})x} \alpha_{ij}|i\rangle\langle j|\,. \label{Eq:StateAfterEnc1d}
\end{align}
By Eq.~\eqref{obs_def} and~\eqref{Eq:ReservoirEvo} the expectation value of a readout observable $O$ takes the form 
\begin{align}\label{Eq:MOutcome}
    \expval{O}_{x}&= \operatorname{Tr}[O U_R\left(\rho(x)\otimes|0\rangle\langle 0|\right)U_R^\dag]\\
    &= \sum_{\omega\in\Omega} a_{\omega}e^{i\omega{x}}\,,
\end{align}
where the Fourier frequencies are given by the differences of the eigenvalues of $H$ 
\begin{equation}\label{Eq:FreqSet}
    \Omega=\{\lambda_j-\lambda_i: i,j = 1, 2, \dots, 2^{n_{a}}\}\,, 
\end{equation}
and the corresponding Fourier coefficients are
\begin{equation}\label{Eq:FParam}
    a_{\omega}=\sum_{i, j \mathrm{\,s. t.\,} \lambda_j - \lambda_i =\omega} \alpha_{ij}\bra{j,0}U_R^\dag{}O U_R \ket{i,0}\,.
\end{equation}
Thus we already see that only the eigenvalues of the Hamiltonian determine the achievable frequencies $\Omega$. Instead, the weighting $a_{\omega}$ of each frequency is determined by the initial state, reservoir dynamics and choice in the observable. 

These observations propagate through to the final function model. More concretely, given a set of observables $\{O_{1}, O_{2}, \dots, O_{M}\}$, we denote the Fourier coefficients resulting from $O_{k}$ as $a_{\omega}^{(k)}$. Then, by Eq.~\eqref{Eq:sysout} the system output can be expressed as a Fourier series of the form
\begin{align}
    f_{\bm{\eta}}{(x)}
    &=\sum_{\omega\in\Omega} b_{\w} e^{i\omega x}\,,
\end{align}
where $b_{\w}=\sum_{k=1}^M \eta_k  a_{\omega}^{(k)}$.

Thus we see that, similar to the Fourier series of a variational quantum circuit, the encoding strategy determines the eigenvalues of the encoding generator, and hence the Fourier frequencies. On the other hand, the Fourier coefficients $b_\omega$ are independent of the input and determined by the dynamics of the reservoir, the trainable weights and the choice of observables. 

In the case of vectorial inputs and a multivariate target function, the expression of the system prediction is easily generalized to a multivariate Fourier series, as shown in Ref.~\cite{schuld2021effect}. For concreteness, we spell this out in Appendix~\ref{AppMultiFourier}. 

{Next, we study the scaling of the number of Fourier frequencies with respect to the number of qubits for different encoding strategies. To make the analysis more tractable,} we {consider the simplest non-trivial scenario of a} {local encodings --} the inputs are encoded on each accessible qubit in parallel. {That is,} the encoding unitary applied on the accessible qubits is separable and takes the tensor product form
\begin{equation}
    U(x)=U_1(x)\otimes U_2(x)\otimes\cdots\otimes U_{n_{a}}(x)\,,
\end{equation}
where $U_k(x)$ denotes the encoding of the $k$'th accessible qubit and is given by
\begin{equation}
    U_k(x)=e^{i H_k x}\,.
\end{equation}
Note that $H_k$ is the generator of the $k$'th accessible qubit, while $H$ in Eq.~\eqref{eq:encoding_def} is the generator of the effective qudit representing all accessible qubits. Let $\lambda^{(k)}_0$ and $\lambda^{(k)}_1$ be the eigenvalues of $H_k$, then, as shown in~\cite{schuld2021effect}, the set of achievable Fourier frequencies can be written as
\begin{align}
    \begin{split}
        \Omega&=\Bigl\{(\lambda^{(1)}_{i_1}+\lambda^{(2)}_{i_2}+\cdots+\lambda^{(n_{a})}_{i_{n_{a}}})-\\
        &(\lambda^{(1)}_{j_1}+\lambda^{(2)}_{j_2}+\cdots+\lambda^{(n_{a})}_{j_{n_{a}}}): i_k,j_k=0,1\Bigr\}
    \end{split}\\
    \begin{split}\label{Eq:Omega}
        &=\Bigl\{(\lambda^{(1)}_{i_1}-\lambda^{(1)}_{j_1})+(\lambda^{(2)}_{i_2}-\lambda^{(2)}_{j_2})+\\
        &\cdots + (\lambda^{(n_{a})}_{i_{n_{a}}}-\lambda^{(n_{a})}_{j_{n_{a}}}): i_k,j_k=0,1\Bigr\}\,,
    \end{split}
\end{align}
where $i_k$ and $j_k$ are the indices for the $k$'th accessible qubit. The set of indices for all accessible qubits $(i_1,i_2,\dots,i_{n_{a}})$ and $(j_1,j_2,\dots,j_{n_{a}})$ correspond to $i$ and $j$ in Eq.~\eqref{Eq:FreqSet} respectively (i.e. $i$ corresponds to $i_1i_2 \cdots i_{n_{a}}$ in binary).

The expression for $\Omega$ in Eq.~\eqref{Eq:Omega} implies that for a fixed number of accessible qubits $n_{a}$ the spacing of eigenvalues of $H_k$ determines the quantity of Fourier frequencies. If the differences between the eigenvalues $(\lambda^{(k)}_{i_k}-\lambda^{(k)}_{j_k})$ are identical for all $k$'s, then the number of achievable frequencies will be significantly smaller, compared to the case when each $H_k$ provides distinguishable differences of eigenvalues. The maximum number of possible distinct non-negative frequencies is $1+ (4^{n_{a}}-2^{n_{a}})/2$. However, often the eigenvalues are distributed such that the number of distinct frequencies is actually lower than this bound.

In the following, we analyse the frequencies generated from the proposed encoding strategies and discuss the controllability of {the trainable weights over} the Fourier coefficients. We start by considering two typical encoding schemes and discuss the quantity of their frequencies, which are further linked to classical simulability.  {Lastly, we note that in general one could relax the local assumption and follow the similar procedure which potentially leads to more complicated analysis to determine the frequency scaling with respect to the number of qubits.} 

\subsection{Encoding strategies}\label{subsec:encoding}
\begin{figure}
    \centering
    \includegraphics[width=0.495\textwidth]{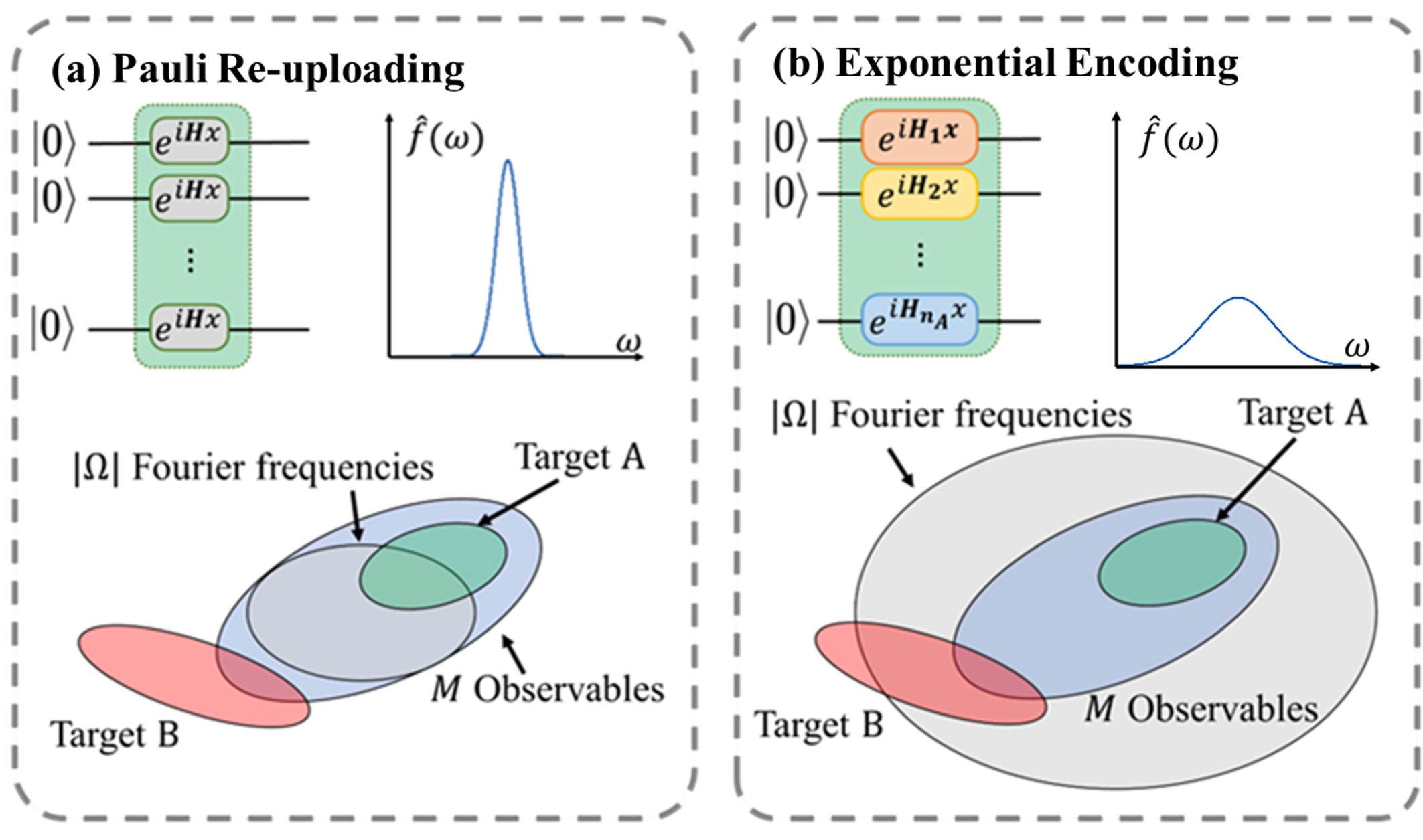}
    \caption{\textbf{Encoding strategies.} The encoding strategy determines the achievable Fourier frequencies of the prediction. Pauli re-uploading and the exponential encoding {(as defined in Sec.~\ref{subsec:encoding})} lead to polynomially and exponentially many frequencies in panels \textbf{(a)} and \textbf{(b)} respectively. Hence, {as shown in the plots of the Fourier transform $\hat{f}(\omega$) of the output $f(\xv)$ against the Fourier frequency $\omega$,} the prediction of a QELM using Pauli encoding has a more concentrated Fourier spectrum, which is more efficiently simulated classically. The exponential encoding, which corresponds to the partial control regime $M<|\Omega|$, allows for a wider range of target functions compared to the Pauli re-uploading, where $M>|\Omega|$ and might offer a quantum advantage.}\label{fig:Encoding_comp}
\end{figure}
\paragraph*{Pauli re-uploading.}
A widely used strategy of encoding is Pauli re-uploading~\cite{sakurai2022quantum,wang2022variational,mujal2021opportunities,sainz2022quantum}, which applies an identical single-qubit Pauli rotation gate on each of the accessible qubits in parallel, as shown in panel (a) of Fig.~\ref{fig:Encoding_comp}. Taking $H_k=\sigma_z/2$ we have
\begin{equation}
    U_k(x)=e^{i\sigma_z x/2},\ \ \forall\ k=1,2,\dots, n_{a} \,.
\end{equation}
As all the generators $\sigma_z/2$ have eigenvalues $\pm1/2$ from Eq.~\eqref{Eq:Omega} we obtain the following achievable Fourier frequencies
\begin{equation}
    \Omega = \{ -n_{a},-n_{a}+1,\dots,0,\dots, n_{a}-1,n_{a}\}\, .
\end{equation}
This is identical to the frequencies of a QNN using Pauli re-uploading~\cite{schuld2021effect,caro2021encodingdependent}. The output function contains all the integer frequencies up to $n_{a}$ and only $(n_{a}+1)$ non-negative frequencies are achievable.

\paragraph*{Exponential Encoding.} 
In order to increase the quantity of frequencies, the exponential encoding has been proposed~\cite{shin2023exponential}. Instead of using the same generator for all accessible qubits, the exponential encoding employs different generators with exponentially scaled spacing of eigenvalues, as illustrated in panel (b) of Fig.~\ref{fig:Encoding_comp}. For the $k$'th accessible qubit, the generator is
\begin{equation}
    H_k=\frac{1}{2}\beta_k \sigma_z\,,
\end{equation}
where $\beta_k = 3^{k-1}$ and $k=1,2,\dots, n_{a}$. By Eq.~\eqref{Eq:Omega}, we obtain the following achievable frequencies:
\begin{align}\label{OmegaExp}
    \Omega =& \biggr\{-\frac{3^{n_{a}}-1}{2}, -\frac{3^{n_{a}}-1}{2}+1, \dots, 1,\\ &0, 1, \dots , \frac{3^{n_{a}}-1}{2}-1, \frac{3^{n_{a}}-1}{2} \biggr\}\,.
\end{align}
This set consists of all the integer frequencies up to $(3^{n_{a}}-1)/2$, totalling $1+(3^{n_{a}}-1)/2$ non-negative frequencies, which is exponential in the number of accessible qubits. 

Next, we demonstrate the predictive power of the two encoding schemes on an artificial classical data set. To generate the dataset we first define a function {$f(x) = \sum_{k=0}^{(3^6-1)/2} a_k \cos({kx}) + b_k\sin({kx})$} where the {$a_k,b_k$}'s are sampled uniformly in $[-1,1]$. We then consider the dataset $\{x_i,f(x_i)\}_{i=1}^{5000}$ where $x_i$'s are placed equidistantly in {$[0,2\pi]$}. {30 percent of this dataset is then randomly set aside as the test data and the model is trained on the rest.}

Two {QELMs} with $n_{a}= 6, n_{h} = 0 $ are trained, one with exponential encoding and another one with Pauli encoding. Both models have the same reservoir, namely a Random Rotation reservoir with $10$ layers of random single-qubit rotations followed by CNOT gates, and in both cases we measure the same subset of randomly chosen Pauli {strings} at the end {without any shot noise and considering infinite statistics}. The training is then performed using an Ordinary Linear Regression model. {We repeat this procedure for 20 random target functions to obtain robust results.}

We expect the exponential encoding to be able to perfectly reconstruct $f$ given enough training data and observables because its frequency spectrum $\Omega=\{0,\cdots, \frac{3^{n_{a}}-1}{2}\}$ completely covers that of $f$. On the other hand, we expect Pauli encoding to fail in approximating $f$ because it can only cover frequencies from $0$ to $6$. 
In Fig.~\ref{fig:comparison} {the mean square error (MSE)} of the two models is shown over the number of observables $M$. We see that even with an exponential number of observables, the Pauli encoding cannot perform well.

Instead, the exponential encoding reaches {achieves an error of zero, i.e., perfectly fits the target function} when we use {more observables than the number of} frequencies in $f${, which is indicated by a green vertical line in Fig.~\ref{fig:comparison}}. This simple example already demonstrates that the exponential encoding, given its superior expressivity, is able to outperform a more naive encoding. Crucially, however, the exponential encoding only leads to better results when there are enough observables. This is discussed in more detail in the following section.

\begin{figure}
    \centering
    
    \includegraphics[scale=0.6]{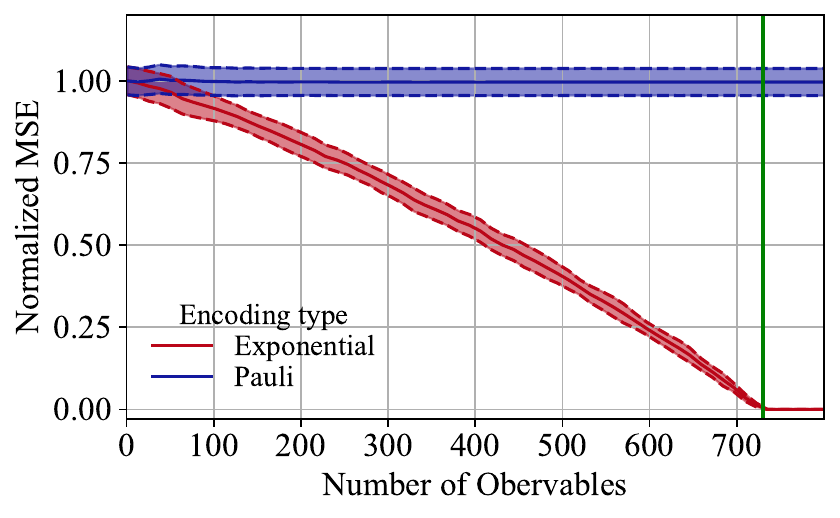}
    
    \caption{\textbf{Comparison between two encoding schemes.} The {mean} performance of two QELMs with exponential and Pauli encoding respectively as a function of the number of observables, {averaged over 20 randomly selected target functions.} {The Mean Squared Error (MSE) evaluated on test data is then normalized to the range $[0,1]$. The dashed lines show the standard deviation of the normalized MSE. The vertical green line shows the number of frequencies expressible by exponential encoding ($3^6 +1= 730$). This line separates two controllability regimes, which are defined and discussed in Sec.~\ref{Sec:Control}.}}
    \label{fig:comparison}
\end{figure}

\subsection{Controllability {over} Fourier coefficients and Expressivity}\label{Sec:Control}

In the previous section we discussed the range of possible frequencies for two different encoding strategies. In particular, we suggested that the exponential encoding strategy could be used to access exponentially many frequencies. However, this does not mean the QELM's prediction can always completely explore the space spanned by those Fourier modes. More specifically, in some cases the Fourier coefficients of the prediction {$b_{\w}$, which are linear combinations of the Fourier coefficients $a_{\w}$ of observables, }{are not linearly independent}. This means the trainable weights $\bm{\eta}$ cannot fully control all the Fourier modes. Here, {we introduce the concept of controllability (of trainable weights over Fourier coefficients of prediction), which represents the scope of QELM's predictive capabilities that can be adjusted by controlling the trainable weights}. We investigate this controllability and discuss the predictive power of QELM in terms of Fourier-expressivity defined below. {We remark that the controllability is strongly related to the Fourier-expressivity (defined below), however, while the Fourier-expressivity is formally defined at the level of prediction as a classical function of inputs, the concept controllability focuses on the training process of QELMs and measures the effect of optimising the parameters for linear regression.}

{To explicitly see the interplay between these two concepts in the context of QELMs with the time-evolution encoding,} {r}ecall Eq.~\eqref{Eq:sysout} and \eqref{Eq:MOutcome}, the prediction of QELM is given by
\begin{align}
    f_{\bm{\eta}}{(x)}&=\sum_{k=1}^M \eta_k \expval{O_{k}}_x \label{Eq:25}\\
    &=\sum_{\omega\in\Omega} b_{\w} e^{i\omega x}\,.\label{Eq:AutoF}
\end{align}
{Here, the Fourier-expressivity can be seen as the number of independent Fourier basis elements, i.e. the degree of freedom of $\vec{b}$, while the controllability can be seen as whether changing $\vec{\eta}$ can access any functions spanned by the Fourier basis.} 

{More formally,} {t}o assess the predictive power, we introduce the dimension of the prediction's function space as a measure of expressivity. This is defined as the cardinality of the smallest set of {any real} functions that form a basis for the model's prediction space.
\begin{definition}[Fourier-expressivity of QELM]\label{def:fourier-expressivity}
Given a QELM as defined above with model prediction $f_{\bm{\eta}}$, let $\mathcal{B}$ be a finite set of {any} real {basis} functions, such that for any trainable weight vector $\bm{\eta}$, $f_{\bm{\eta}}$ can be written as a linear combination of the elements of $\mathcal{B}$, i.e.
\begin{equation}\label{eq:Bdef}
    \{f_{\bm{\eta}}(\cdot): \bm{\eta} \in \mathbb{R}^M\}\subseteq span(\mathcal{B})\,.
\end{equation}
We define the Fourier-expressivity of $f_{\bm{\eta}}$ as the cardinality of the smallest possible $\mathcal{B}$, i.e., 
\begin{equation}
    \mathcal{F}\left[f_{\bm{\eta}}\right]:=\operatorname{min}|\mathcal{B}|\, ,
\end{equation}
where the minimization is over all $\mathcal{B}$ satisfying Eq.~\eqref{eq:Bdef}.
\end{definition}
On the one hand, the output $f_{\bm{\eta}}$ is a linear combination of $M$ outcomes{, as shown in Eq.~\eqref{Eq:25}}, which implies that the number of observables $M$ upper bounds the Fourier-expressivity $\mathcal{F}$. On the other hand, the prediction can be written as a Fourier series of $|\Omega|$ frequencies. Since the terms of different frequencies are orthogonal, $\mathcal{F}$ is also upper bounded by $|\Omega|$. The locality of the measurements also limit the Fourier-expressivity. Assuming $n_{o}$ qubits are measured, then any observable can be written as a linear combination of $4^{n_{o}}$ Pauli strings. This implies that no more information can be extracted than that obtained from measuring $4^{n_{o}}$ observables. We further remark that the number of frequencies $|\Omega|$ cannot be more than the number of possible distinct differences of eigenvalues which is upper bounded by $4^{n_{a}}$. 
These arguments lead us to the following Theorem, which is formally proved in App.~\ref{App:UpperBound}.

\begin{theorem}[Upper bound of QELM's Fourier-expressivity]\label{Observa1}
Consider a QELM, as defined above, with model prediction $f_{\bm{\eta}}$. Let $M$ be the number of observables, $\Omega$ be the set of achievable frequencies, and $n_{o}$ be the number of measured qubits, then
    \begin{align}\label{Claim1Eq}
        \mathcal{F}\left[f_{\bm{\eta}}\right] &\leq\operatorname{min}\{M, |\Omega|, 4^{n_{o}}\} \,,
    \end{align}
where $|\Omega|\leq 4^{n_{a}}$.
\end{theorem}

As the upper bound $4^{n_{o}}$ scales exponentially in number of measured qubits {$n_{o}$}, while the number of observables $M$ is in practice polynomially large in the system size {(i.e. the number of total qubits $n>n_{o}$)}, in what follows we will assume that {$n$ is polynomially large in $n_{o}$, and hence} $M$ is always less than $4^{n_{o}}$. {In other words, we assume that in practice the number of observables measured is always less than $4^{n_{o}}$, such that the information obtained from the measurement is never over-completed (see App.~\ref{App:UpperBound} for more details).} In this case, the upper bound reduces to 
\begin{equation}
        \mathcal{F}\left[f_{\bm{\eta}}\right] \leq\operatorname{min}\{M, |\Omega|\} \label{Eq:RankIneq} \, . 
\end{equation}

We are now left with two questions: \textit{first, when does the inequality in Eq.~\eqref{Eq:RankIneq} saturate? Second, what can we expect if $M\geq|\Omega|$ and if $M<|\Omega|$?}
For the first question, we claim that in practice this upper bound is always tight. We observe this numerically in Section~\ref{SecRoleR} for all the practical reservoirs, including integrable reservoirs, chaotic reservoirs, and {Haar random} reservoirs.
Next, we address the second question by comparing the following two regimess:\\
\paragraph*{Full control regime - $M\geqslant{}|\Omega|$.}
As long as we choose {a} non-trivial {set of} observables{, e.g. a set of random observables,} such that the functions $\expval{O_{k}}_x$ {of input $x$} with $ 1\leq k\leq M $ are {linearly} independent of each other, {i.e. any observable's expectation cannot be written as a linear combination of the orthers}, we have full control over the Fourier coefficients $b_\omega$. 

{We remark that whether those expectation values are linearly independent also depends on the reservoir. For reservoirs with rich dynamics, as discussed in Sec.~\ref{SecRoleR}, one could choose a set of orthogonal observables, e.g. a set of Pauli observable, to obtain the full control over Fourier coefficients. However, in general a set of orthogonal observables cannot guarantee the full control, since their projections onto the subspace of reservoir states after the reservoir evolution are not always linearly independent.} 

In {the full control} regime, we have
\be\label{Eq:Control1}
    \mathcal{F}\left[f_{\bm{\eta}}\right]=|\Omega|
\ee
and
\be
    \{f_{\bm{\eta}}(x): \bm{\eta} \in \mathbb{R}^M\}=\operatorname{span}(\{e^{i\omega x}:\omega\in \Omega\}) \,,
\ee
meaning that the model can learn any target function whose Fourier {frequencies} are within $\Omega$. 

{In the toy example studied in the previous section, the full control regime corresponds to the right hand side of the green line in  Fig.~\ref{fig:comparison}. As shown by the red curve, if the encoding strategy provides sufficient Fourier frequencies, the number of observables is larger than the number of frequencies, and the reservoir has rich dynamics, such that the expectation values are linearly independent, then the 0 test error can always be achieved.}
    
However, in practice, the number of {observables} implemented can only scale at most polynomially, i.e. $M \in \OC(\mathrm{poly}(n))$. {Otherwise, one will need exponentially many measurement shots to obtain precise enough results even from a single measurement involving exponentially many observables. On the other hand, the training of exponentially many weight parameters will be inefficient}. It follows that the number of frequencies scales at most polynomially, i.e. $|\Omega| \in \OC(\mathrm{poly}(n))$. We recall that the prediction of a QELM using Pauli re-uploading encoding scheme results in a polynomial number of frequencies. As we discuss in Section~\ref{SecSurro}, such a QELM can be classically surrogated, and hence does not yield quantum advantage. \\

\paragraph*{Partial control regime - $M<|\Omega|$.}

In this regime, the number of observables or the number of weights is less than the number of Fourier frequencies. Hence, the Fourier coefficients $b_\omega$ are not independent, and we can only partially control the model prediction. Suppose the $\expval{O_{k}}$'s are independent of each other, then we have
\begin{equation}\label{Eq:Control2}
    \mathcal{F}\left[f_{\bm{\eta}}\right]=M 
\end{equation}
and
\begin{equation}
    \{f_{\bm{\eta}}(\cdot): \bm{\eta} \in \mathbb{R}^M\}\subset\operatorname{span}(\{e^{i\omega x}:x\in \Omega\}) \, .
\end{equation}
That is, some Fourier series from the spectrum $\Omega$ cannot be predicted by the QELM as the degrees of freedom are limited by the number of observables. {As shown in Fig.~\ref{fig:comparison}, if the number of observables is less than the number of frequencies, than the test error is always non-zero for randomly generated target functions.}

This regime includes the QELMs using exponential encoding strategies such that $|\Omega| \in \OC(\mathrm{exp}(n))$, that are potentially classically non-surrogatable (as discussed later in Section~\ref{SecSurro}) and so might offer a quantum advantage. 

{Lastly, we note that even though dubbed ``Fourier''-expressivity the formal definition (i.e. Definition~\ref{def:fourier-expressivity}) is essentially defined independent of the choice of function basis. In this work, we focus on the Fourier basis to analyse the expressivity and the controllability since this is a natural choice as a result of the ``time-evolution'' embedding. In general, other decompositions using different function basis types could be considered especially for the data embeddings beyond the ``time-evolution'' approach. For example, one could potentially consider $\{ \rho(\xv) \}_{\xv \in \XC}$ is spanned by a set of observables $ \{ \tilde{O}_i \}_{i=1}^{D_o}$ i.e. for all $\xv\in \XC$, $\rho(\xv) = \sum_{i=1}^{D_o} a_i(\xv) \tilde{O}_i$ for some real coefficients $\{ a_i(\xv) \}_i$. Then, one could argue that the natural set of observables to be picked for QELMs is also $\{ \tilde{O}_i \}_i$, which lead to the model prediction of the form $f_{\vec{\eta}}(\xv) = \sum_{i=1}^{D_o} \eta_i \Tr[\tilde{O}_i \rho(\xv)] = \sum_{i,j} \eta_i a_j(\xv) \Tr[\tilde{O}_i \tilde{O}_j]$. However, further investigation of such decomposition regarding the expressivity and controllability is beyond scope of the work.}

\subsection{Role of reservoir}\label{SecRoleR}
\begin{figure*}
    \centering
    
    \includegraphics[width=\textwidth]{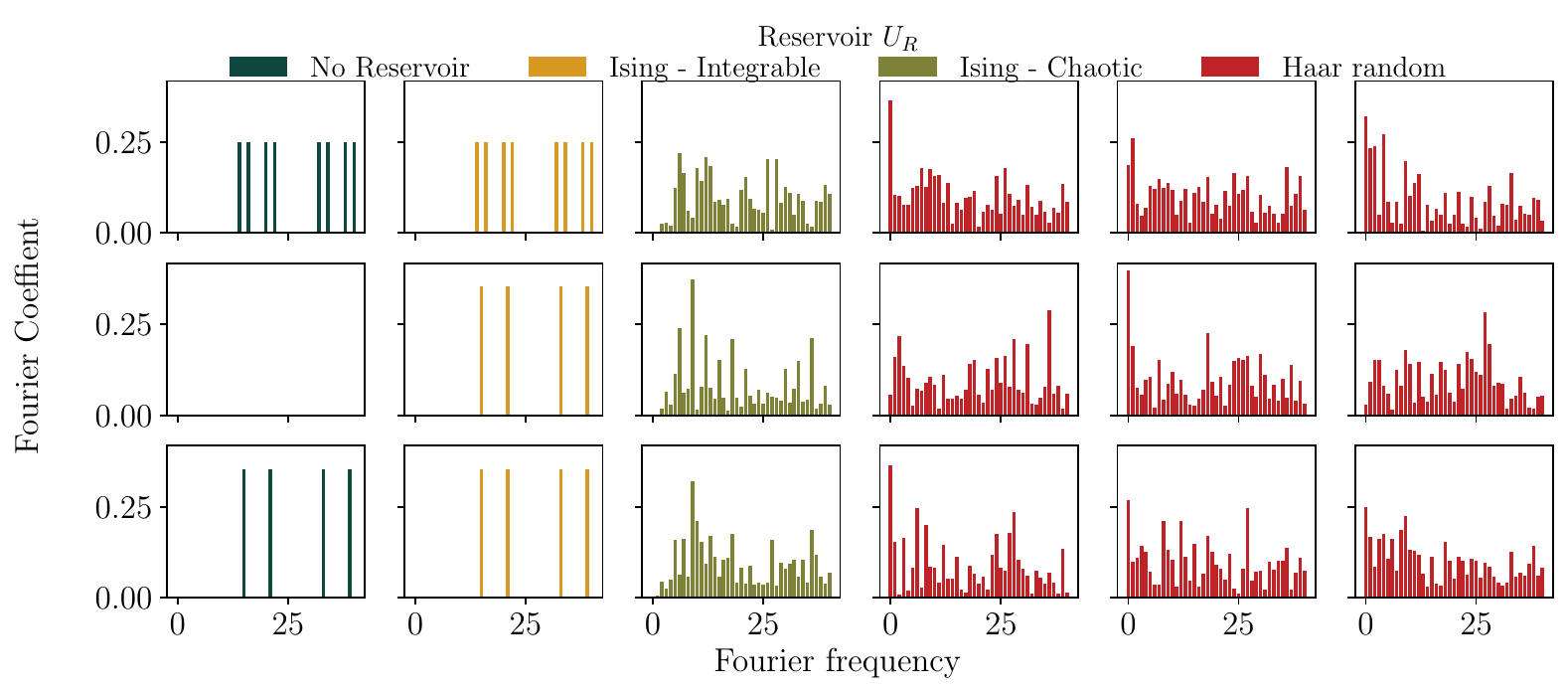}
    
    \caption{\textbf{Exemplary Fourier spectra of a QELM with $n_{a}=4$ and $n_{h}=4$.} The Fourier coefficients of three observables, $O_1=X^{\otimes 4}\otimes \IC$, $O_2=X\otimes Y \otimes Y\otimes Z \otimes \IC$ and $O_3=O_{rand}\otimes \IC$ (where $O_{rand}$ is a random Hermitian in the accessible space), resulting from six $reservoirs$: no reservoir, integrable Ising model ($J=-1,\ B_x=0,\ B_z=1$), chaotic Ising model ($J=-1,\ B_x=0.7,\ B_z=1.5$) and three different Haar random reservoirs. Each row (from top to bottom) corresponds to an observable (from $O_1$ to $O_3$ respectively), while each column corresponds to a reservoir unitary indicated on the top. {We remark that the motivation of showing plots for three different Haar random reservoirs is to avoid special cases.}}
    \label{fig:Spectrum}
\end{figure*}
In the regime of partial control, i.e. the regime where a quantum advantage might be possible, the number of observables $M$ is much smaller than the number of frequencies $|\Omega|$. Hence, the dimension of prediction's function space is $M$ and each readout $\expval{O_{k}}$ can be viewed a linearly independent basis function. In the previous section, we showed that exponentially many frequencies can be achieved. However, it might be the case that despite a wide range of frequencies being enabled by the encoding, the Fourier coefficients are sparse and close to zero. This may limit the Fourier-expressivity of the model and make it prone to classical simulation. In this part, we look into the Fourier spectra of these basis functions and study the richness of their Fourier modes --the proportion of non-zero Fourier coefficients. We find that the dynamics of the reservoir significantly affects the Fourier spectra and the richness. 

For concreteness, we consider a practically efficient measurement strategy, namely measuring Pauli strings, and assume that only the accessible qubits will be measured. 
We then compare four types of reservoir settings: no reservoir, the integrable Ising model, the chaotic Ising model, and a typical Haar random unitary.
In particular, we consider the 1D Ising model,
\begin{equation}\label{Eq:HIsing}
    H_{\operatorname{Ising}}=J\sum_{i=1}^{n-1}Z_i{}Z_{i+1}+B_z\sum_i^{n}Z_i+B_x\sum_i^n{}X_i\,,
\end{equation}
{where $X_i$ and $Z_i$ are Pauli X and Z operators for the \mbox{$i$-th qubit} respectively} and, following Ref.~\cite{geller2022quantum}, {we} use the parameters 
\begin{itemize}
    \item Integrable regime: $J=-1,\ B_x=0,\ B_z=1$
    \item Chaotic regime: $J=-1,\ B_x=0.7,\ B_z=1.5$\,.
\end{itemize}

In Fig.~\ref{fig:Spectrum} we provide examples of the Fourier spectrum induced by these reservoirs for three different Pauli observables. We observe that the integrable Ising model leads to a spectrum that only weakly deviates from using no reservoir. In contrast, the chaotic Ising model and {Haar random} reservoirs lead to more complex, anti-concentrated Fourier spectra basis functions. In Appendix~\ref{stats}, we analytically compute the expected spectrum for a Haar random reservoir and use this to explain the anti-concentration of the frequencies observed in Fig.~\ref{fig:Spectrum}. As the chaotic Ising model and a typical Haar random unitary are both models of scramblers it is intuitive that their effect on the observed frequencies, as observed here, is similar~\cite{roberts2017chaos, holmes2021barren, geller2022quantum}.

To quantify how many of the achievable frequencies --determined by encoding strategy-- have a non-zero coefficient, we plot the richness of Fourier modes with respect to the number of accessible qubits for different reservoirs. Given a reservoir, the richness is defined as the proportion of non-zero Fourier coefficients averaged over all the Pauli strings in accessible space, i.e.
\begin{equation}
    R:=\frac{1}{4^{n_{a}}}\sum_{k=1}^{4^{n_{a}}}{|\{\omega\in\Omega: a_\omega^{(k)}\neq0\}|}\,,
\end{equation}
where the Fourier coefficients $a_\omega^{(k)}$ correspond to the $k$'th element of $\big\{ \{I,X,Y,Z\}^{\otimes n_{a}}\otimes\IC \big\}$.

As shown in Fig.~\ref{fig:Richness}, without a reservoir and for a reservoir governed by the integrable Ising model, the richness $R$ of the output model decreases exponentially in the number of accessible qubits $n_{a}$. On the other hand, the richness of the models using {Haar random} and chaotic Ising reservoirs saturate at a constant value close to $1$. Notably, the chaotic Ising model leads to same richness $R$ as a Haar random reservoir. 

\begin{figure}
    \centering
    
    \includegraphics[width=0.45\textwidth]{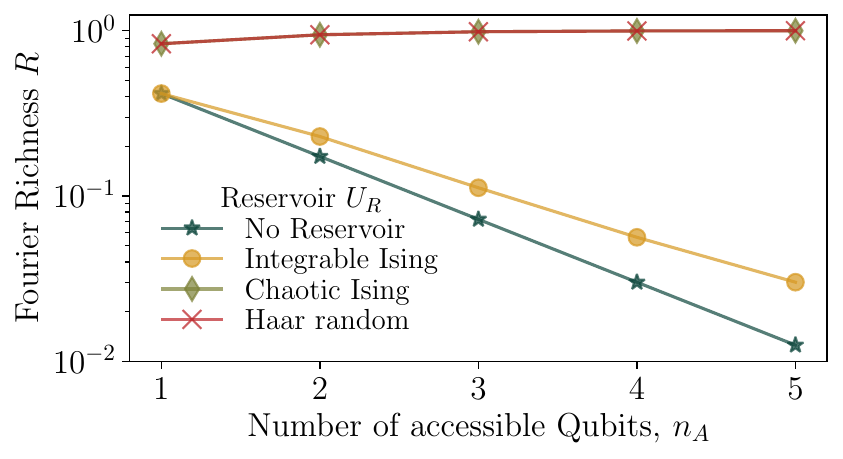}
    
    \caption{\textbf{Richness of Fourier modes.} The richness of Fourier modes, i.e. the proportion of non-zero Fourier coefficients averaged over all the Pauli observables in the accessible space $\{I,\ X,\ Y,\ Z\}^{\otimes n_{a}}\otimes \IC$, are plotted on a logarithmic scale against the number of accessible qubits $n_{a}$ (the number of hidden qubits is fixed to $n_{h}=4$). } 
    \label{fig:Richness}
\end{figure}

\subsection{Classical surrogates and the potential for a quantum advantage}\label{SecSurro}

So far we have focused on the Fourier-expressivity and controllability of QELMs. While a firm understanding of these aspects is crucial, for QELMs to have a practical usefulness one also need to consider their potential advantage over classical computers. In particular, studying the conditions under which QELMs can be dequantized gives us a practical guideline of where to look for a quantum advantage.

One popular approach to dequantize quantum machine learning models is via a classical surrogate. Under this approach, a purely classical model (known as a surrogate) is built to mimic a respective QML model. Intuitively, this can be achieved by exploiting the known Fourier structure of model predictions generated by these QML models. We now briefly discuss two methods for building classical surrogates that were originally developed for quantum neural networks and quantum kernels and discuss how they can be applied to QELMs.

The first method proposed in Ref.~\cite{schreiber2022classical} is to directly {approximate} {outputs of quantum models using} Fourier series on a classical computer{, for which the quantum computer might be needed to determine the Fourier coefficients beforehand}. This method can only be applied to quantum models whose outputs containing polynomially many Fourier frequencies. {It relies explicitly on the fact that with polynomial frequencies the form of the model prediction (i.e. the Fourier series) can be represented efficiently with classical computers.} To determine the associated Fourier coefficients, {one first determine the frequencies from the encoding scheme and then learn the coefficients corresponding to these frequencies by performing a regression task whose training data consists of random inputs and their respective outputs after fed into the quantum circuit. Then, one obtains a proper approximation of the quantum circuit and access to the quantum computer is no longer needed for further evaluation of new input data.}

Since the model predictions of the QELM take the same Fourier form as the QNNs, this surrogate method is directly applied to the QELM framework with classical data. {That is}, the QELM with a Pauli embedding (and hence only a polynomial number of frequencies) is classically surrogatable. 
{Particularly, given a quantum circuit with Pauli encoding (or, any embedding with polynomial frequencies), we already know that the output has a Fourier representation and we can also find the frequencies given the encoding gates, so the only unknown parameters are the Fourier coefficients.} If the reservoir and measurements are classically simulable, the model becomes surrogatable without any need for the quantum computer. On the other hand, for more complex reservoirs and measurements, a polynomial number of {queries} must first be taken on the quantum computer {to obtain the input-output paris for training and then the coefficients can be optimised by solving a system of linear equations. In this method, the number of data points for which we evaluate the circuit, is at least equal to the number of frequencies in the quantum model.}

In the second method, known as Random Fourier Features (RFF)~\cite{landman2022classically, sweke2023potential}, the key idea is to sample Fourier frequencies weighted with the associated coefficients to build a classical surrogate. This is in contrast to the first method where the whole Fourier series is used by the surrogate. This classical surrogate performs at least as well as the original quantum model on the training dataset, but the surrogate could perform worse on unseen data if the original quantum model has an inductive bias that better aligns with a target function. One strength of the RFF method is that it could even classically surrogate quantum models with exponentially many frequencies if Fourier coefficients concentrate over polynomial regions. However, if the coefficients are well spread throughout the whole spectrum range, it is inefficient to surrogate by RFF. Similarly, the RFF method also applies to the QELMs. 

Thus to achieve a quantum advantage with a QELM one should use an encoding with an exponential number of frequencies and reservoirs/observables that ensure the weights of these frequencies are anti-concentrated. In Fig.~\ref{fig:Spectrum} we show that a chaotic Ising or Haar random reservoir with Pauli observables give rise to such anti-concentration. We further support these numerics with an analytic calculation of variances for the case of the random reservoir {sampled from 2-design} in Appendix~\ref{stats}. These results show that the coefficients are well spread and thus suggest that some unitaries in the Haar random family are not surrogatable via the RFF method.  

However, a QELM being non-surrogatable by RFF, does not mean the model cannot be surrogated by other approaches. As discussed in Appendix~\ref{stats}, for a QELM with the Haar random reservoir there exists a trivial classical surrogate model {for the large number of qubits}, due to {the randomness from a Haar-random reservoir}. {In particular, we show that \textit{guessing zero} regardless of the input data is already a good classical surrogate for this scenario i.e. $f_{\rm Surr}(\vec{x}) = 0$.} 

We also demonstrate in the next section that the highly Haar-expressive reservoir {(i.e. a reservoir with a unitary sampled from an ensemble whose second moment closely resembles that of a Haar distribution, see Section~\ref{sec:expressibility_concentration})} leads to exponential concentration which causes a QELM to generalize poorly. Therefore, a QELM with a Haar random reservoir will be useless. 

More importantly, having no classical surrogate does not automatically guarantee a quantum advantage. Thus, one may ask the following key question: \textit{Is there any room and hope for a quantum advantage with QELMs?}

We provide a contrived example that a QELM can provably achieve an exponential quantum advantage over any classical model, {assuming the widely believed classical hardness of the discrete logarithm problem (DLP).} Our example is heavily based on Ref.~\cite{liu2021rigorous} where the authors prove an advantage of a quantum kernel method to solve a particular classification task. {Here, we consider an original discrete input space $\XC = \{1, 2, ..., p \}$ with a large prime number $p$ and the bits required to represent the prime number $n = \log_2(p)$. Given some integer $g$, it is a widely believed conjecture that $\phi_{\rm log}(\xv) = \log_g(\xv)$ cannot be computed efficiently (in $n$) with classical computers~\cite{odlyzko2000discrete}. Indeed, this is the core argument to show the learning advantage. That is, we consider the classification task where the input data $\xv$ in the original space $\XC_{\rm origin}$ appear to be randomly labelled but, after taking the logarithm function, they are linearly separated in this log space i.e., $\XC_{\rm log} = \{\phi_{\rm log}(\xv) = \log_g(\xv) \; ; \; \forall \xv \in \XC_{\rm origin} \}$. More precisely, given some $s \in \XC_{\rm random}$, the labels $y$ can be expressed as
\begin{align}
    y = \left\{\begin{array}{ll}+1,&\,\text{if}\,{{\rm{log}}}_{g}x\in [s,s+\frac{p-3}{2}],\\ -1,&\,\text{else}\,,\end{array}\right.
\end{align}

Ref.~\cite{liu2021rigorous} rigorously shows that the classical hardness of solving DLP also implies the hardness of solving this learning problem. On the other hand, quantum computers are well known to solve DLP (i.e., computing $\log_g(\xv)$) using Shor's algorithm~\cite{shor1994algorithms}, which is translated to solve this learning task. Particularly, there exist a quantum kernel-based model to solve the learning task with the model predictions of the form~\cite{liu2021rigorous} }
\begin{align}
f(\vec{x}) = \sum_{{i}=1}^{D} \alpha^*_i \Tr[\rho(\vec{x}_i) \rho(\vec{x})] \;,
\end{align}
where $\rho(\xv) = U_{\rm shor}(\vec{x}) \rho_0 U^{\dagger}_{\rm shor} (\vec{x})$ with {$U_{\rm shor}(\xv)$ being the data-embedding constructed using Shor's algorithm as a subroutine,} and $\alpha^*_i$ are the optimal coefficients obtained from solving the support vector machine.

For the condition to translate this into the QELM framework, we can equate both model predictions. That is, we have
\begin{align}
f(\vec{x}) = \sum_{k=1}^{D} \alpha^*_k\Tr[\rho(\vec{x}_k) \rho(\vec{x})] = \sum_{k=1}^M \eta^*_k \Tr [ O_{k} \rho(\vec{x}) ] \;.
\end{align}
One simple choice to achieve this is to take $\alpha^*_k  = \eta^*_k $, $M=D$, and $\rho(\vec{x}_k) =  O_{k} $ for all $ k=1,\cdots, D$. This is a highly contrived example, massaging the work of Ref.~\cite{liu2021rigorous} into the quantum extreme learning formalism, but it nonetheless shows that a quantum advantage can in theory be achieved with a QELM. 

\section{Exponential Concentration in QELMs}\label{trainability_analysis}
In this section, we discuss a fundamental limitation of QELMs. We demonstrate that the prediction of a QELM becomes input-independent when the observables' expectation values are exponentially concentrated in the system size. Similar to the exponential concentration in quantum kernel methods~\cite{thanasilp2022exponential}, this is shown for four different scenarios for QELMs: Haar-expressivity of the model's unitaries (Sec.~\ref{sec:expressibility_concentration}), highly-entangled states (Sec.~\ref{sec:entanglement_concentration}), global observables (Sec.~\ref{sec:global_concentration}) and noisy evolution (Sec.~\ref{sec:noise_concentration}). 

\subsection{Definitions and its consequences}\label{Sec:trainability}
In general, ELMs are guaranteed to be trainable due to the convexity of their loss landscape. This also extends to their quantum counterparts, where training is performed via a linear regression of the outcome observables. However, while in classical systems the exact readout can be obtained, in quantum systems physical quantities can only be estimated through repeated measurements, owing to the statistical nature of such systems. In order to efficiently estimate the expectation values of observables, then, one must be able to approximate its value enough precisely from at most polynomially many measurements.  In simple terms, we would like our QELMs to be able to recognise, through a precise estimation of the output observables, that a distinguishable input was fed into the model. Hence, the expectation values of the observables are expected to be \textit{input-dependent}.

However, the four sources of concentration, which we discuss in the following sections, can result in an arbitrary observable concentrating exponentially with the system size towards an input-independent value, independent of the input. We differentiate between \textit{probabilistic} and \textit{deterministic} exponential concentration.

\begin{definition}[Probabilistic exponential concentration]\label{def_prob_exp_concentration}
Consider a quantity $Q(\vec{\theta})$ that depends on some variable $\vec{\theta}$ which can be estimated from an $n$-qubit quantum computer as an expectation value of some observable. $Q(\vec{\theta})$ is said to probabilistically exponentially concentrate around an input-independent value $\mu$ if
\be\label{eq:def_exp_concentration}
\mathrm{Pr}_{\vec{\theta}}[|Q(\vec{\theta})-\mu|\geq\delta] \leq \dfrac{\beta}{\delta^2}\hspace{2pt},\;\; \beta \in \OC(1/b^{n}) \;,
\ee
for some $b > 1$. That is, the probability that $Q(\vec{\theta})$ deviates from $\mu$ by a small amount $\delta$ is exponentially small for all $\thv$.
\end{definition}
We note that in Eq.~\eqref{eq:def_exp_concentration}, by applying Chebyschev's inequality, $\beta$ can readily be associated to  $\Var_{\vec{\theta}}[Q(\vec{\theta})]$ with $\mu=\Ebb_{\vec{\theta}}[Q(\vec{\theta})]$. As a matter of fact, one often bounds the variance to assess the degree of concentration for a given observable.

\begin{definition}[Deterministic exponential concentration]\label{def_determ_exp_concentration}
Consider a quantity $Q(\vec{\theta})$ that depends on some variable $\vec{\theta}$ which can be estimated from an $n$-qubit quantum computer as an expectation value of some observable. $Q(\vec{\theta})$ is said to deterministically exponentially concentrate around an input-independent value $\mu$ if its distance from it is bounded by an exponentially small value
\be\label{eq:determ_exp_concentration}
|Q(\vec{\theta})-\mu| \leq \beta \hspace{2pt},\;\; \beta \in \OC(1/b^n) \;,
\ee
for some $b >1$. That is, $Q(\vec{\theta})$ does not deviate from $\mu$ for more than $\beta$ for all $\theta$.
\end{definition}

Definitions~\ref{def_prob_exp_concentration},~\ref{def_determ_exp_concentration} are rather general and can be applied to many instances. In the context of QELM, we will consider $\vec{\theta}=\xv$ or $\thv = U_R$ and $Q(\thv)=\expval{O}_{\xv}$, where $O$ can be any one of the observables out of the set used for training. In this case, the probability in Eq.~\eqref{eq:def_exp_concentration}, as well as variance and mean of $\expval{O}_{\xv}$, would be taken over the set of all inputs $\XC$ or over an ensemble from which $U_R$ is drawn.
In practice, each observable can only be estimated by performing $N_{\text{meas}}$ measurements. This leads to a finite-shot error $\epsilon \in \OC(1/\sqrt{N_{\text{meas}}})$. Now, if the observable is found to be exponentially concentrated, then Eq.~\eqref{eq:def_exp_concentration} indicates that in order to faithfully recognise $\expval{O}_{\xv}$ from $\mu$, one would need a resolution $\alpha \in \OC(1/b^n)$. This implies that the error generated by the finite sampling scheme should be less than the resolution, yielding $N_{\text{meas}} \geq b^{n/2}$. Such a dependence on the system size $n$ makes the observable estimation highly inefficient.

In the regime of exponential concentration, training the QELM will still be possible, as the model does not require training on the quantum device and is instead trained through classical convex optimization. However, by employing only a polynomial number of shots to approximate the observables, the model will not learn on input-dependent observables and so it will not be able to recognise the actual input from a completely random one. Hence, we remark that the model's trainability will not suffer from exponential concentration, but its generalization capabilities will. If the model is trained on input data that cannot be distinguished from random data, then the prediction on new data will in turn be random. For a more formal discussion of these ideas see Appendix~\ref{appendix_prob_theory_refresher}.

\subsection{Haar-expressivity-induced concentration}\label{sec:expressibility_concentration}
Loosely speaking, one can define the Haar-expressivity of an ensemble of unitaries of dimension $d$ as a measure of the extent to which the ensemble covers the unitary group $\UC(d)$.
Given a choice of unitary in the encoding strategy of QELM, we can straightforwardly define the unitary ensemble spanned by the encoding unitary

\be\label{eq:encoding_ensemble}
  \Ubb_{\xv} = \{U(\xv):\xv \in \XC\}\,.
\ee
To quantitatively measure the Haar-expressivity of $\Ubb_{\xv}$, one can define the second moment of the Haar distribution in dimension $d_{a}=2^{n_{a}}$ associated to the accessible space
\be
    \VC_{\text{Haar}}(\cdot) = \int_{\UC(d_{a})}\,d\mu(V) V^{\otimes 2}(\cdot)(V\ad)^{\otimes 2}\,,
\ee
{where unitary $V$ is sampled from Haar measure on the unitary group $\UC(d_{a})$ of dimension $d_a$.}

A $2$-design is an ensemble for which its distribution agrees with that of the Haar measure up to the second moment.
Hence, the following super-operator defines the distance between $\Ubb_{\xv}$ and a $2$-design ensemble
\be\label{expressibility_superop}
    \AC_{\Ubb_{\xv}}(\cdot) := \VC_{\text{Haar}}(\cdot) - \int_{\Ubb_{\xv}} \,dU(\xv) U(\xv)^{\otimes 2}(\cdot)(U(\xv)\ad)^{\otimes 2}\,. 
\ee
The Haar-expressivity of the ensemble $\Ubb_{\xv}$ is by how close it is to a $2$-design. So taking, for instance, the diamond norm of Eq.~\eqref{expressibility_superop} is a good measure of Haar-expressivity~\cite{holmes2021connecting}
\be
\varepsilon_{\diamond}^{\xv} := \norm{\AC_{\Ubb_{\xv}}}_{\diamond}\,,
\ee
where the superscript underlines the fact that in general the Haar-expressivity will depend on the input ensemble $\XC$. Alternatively, one could also use the $1$-norm $\varepsilon_1^{\xv}:=\norm{\AC_{\Ubb_{\xv}}}_1$, however for most results we will rely on the diamond norm.

One can show that the encoding unitary can potentially generate exponential concentration with respect to the number of accessible qubits $n_{a}$. Let us formalize this in the following Theorem.
\begin{theorem}[Encoding Haar-expressivity-induced concentration]\label{expressibility_thm1}
    Consider the expectation value of an arbitrary observable as defined in Eq.~\eqref{obs_def} and data-dependent unitary ensemble introduced in Eq.~\eqref{eq:encoding_ensemble}. Then we have that
    \be\label{Eq:Conc_EncodingExpress}
    \mathrm{Pr}_{\xv}[|\expval{O}_{\xv} - \Ebb_{\xv}[\expval{O}_{\xv}]| \geq \delta ] \leq \dfrac{G(\varepsilon_{\diamond}^{\xv})}{\delta^2}\,,
    \ee
    where $\Ebb_{\xv}[\cdot]$ is the expectation taken over $\Ubb_{\xv}$ and the term on the RHS is:
    \be
    G(\varepsilon_{\diamond}^{\xv}) = \dfrac{\big(\Tr[\widetilde{O}_{\Lambda}]^2+\Tr[\widetilde{O}_{\Lambda}^2]\big)}{2^{n_{a}}(2^{n_{a}}+1)} + \varepsilon_{\diamond}^{\xv}\norm{\Lambda(O)}_{\infty}^2\,,
    \ee  
    where $\widetilde{O}_\Lambda = \Tr[(\IC\otimes\ketbra{0}{0})\Lambda(O)]$ with $\Lambda(\cdot)=U_{R}^\dag (\cdot)U_{R}$.
\end{theorem}
A proof of this theorem is provided in Appendix~\ref{appendix_exp_concentration}. Note that, generally, for Pauli observables acting on the accessible space we have $\Tr[O^2] = 2^{n_{a}}$ and so the first term in the definition of $G(\varepsilon_{\diamond}^{\xv})$ is still exponentially small in the number of accessible qubits. Theorem~\ref{expressibility_thm1} indicates that if the encoding ensemble is exponentially close to a 2-design such that $\varepsilon_{\diamond}^{\xv}$ is exponentially small, then the observable will exponentially concentrate towards its expectation value.

We provide also numerics showcasing the phenomenon. We consider a $L$-layered ansatz 
\be\label{numerics_ansatz}
U(\xv) = \prod_{l=1}^LU_l(\xv)W_l
\ee
where $U_l(\xv)=U_{l,1}(\xv)\otimes\dots\otimes U_{l,n_{a}}(\xv)$ are separable unitaries made up of single-qubit rotations, and $W_l$ are a series of entangling gates between neighboring qubits. In Fig.~\ref{fig:encoding_expressibility_numerics}, we show that exponential concentration appears as the ansatz depth $L$ increases, and therefore approaches a $2$-design.
\begin{figure}
    \centering
    \includegraphics[scale=0.6]{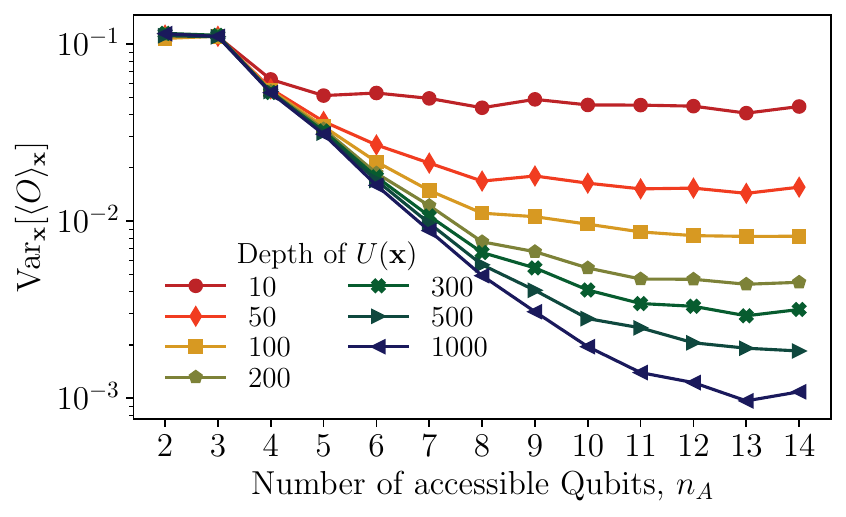}
    \caption{\textbf{Encoding Haar-expressivity-induced concentration.} Variance of the observable $Z_{1}Z_{2}$ over a set of inputs uniformly sampled from $[-\pi,\pi]$, as a function of the number of accessible qubits $n_{a}$ and for different depths of the encoding unitary as defined in Eq.~\eqref{numerics_ansatz}.}
    \label{fig:encoding_expressibility_numerics}
\end{figure}

Similarly to the encoding strategy, we can characterize the exponential concentration stemming from the reservoir unitary evolution. As a matter of fact, in Section~\ref{SecRoleR} we showed that more complex reservoir dynamics such as the chaotic Ising model and a Haar random unitary can provide a richer Fourier spectra, potentially providing quantum advantage for QELM. Now, we would like to investigate if the scrambling degree of a reservoir unitary poses a limitation in terms of exponential concentration.

Suppose that the reservoir unitary $U_{R}$ is drawn from an ensemble $\Ubb_{R}$. Then, 
similarly to Eq.~\eqref{expressibility_superop}, we can define its Haar-expressivity super-operator as the distance between the Haar measure over the unitary ensemble of dimension $d$ and the average over $\Ubb_{R}$:

\be
    \AC_{\Ubb_{R}}(\cdot) := \VC_{\text{Haar}}(\cdot) - \int_{\Ubb_{R}} \,dU_{R} U_{R}^{\otimes 2}(\cdot)(U_{R}\ad)^{\otimes 2} \,. 
\ee
Then the diamond norm $\varepsilon_{\diamond}^{R}:=\norm{\AC_{\Ubb_{R}}}_{\diamond}$ measures the Haar-expressivity of the reservoir. Depending on whether the ensemble $\Ubb_{R}$ forms a 2-design (such that $\varepsilon_{\diamond}^{R} = 0$), there may be a tendency for observables to concentrate around their average. Let us formalize this in a Theorem.

\begin{theorem}[Reservoir Haar-expressivity-induced concentration]\label{expressibility_thm2}
    Consider a reservoir evolution $U_R \in \Ubb_R$. Consider the expectation value of an arbitrary Hermitian observable as defined in Eq.~\eqref{obs_def}. Then we have that
    \be\label{Eq:Conc_Res_Express}
    \mathrm{Pr}_{U_R}[|\expval{O}_{\xv} - \Tr[O]/d| \geq \delta ] \leq \dfrac{G(\varepsilon_{\diamond}^{R})}{\delta^2}\,,
    \ee
    where 
    \be
    G(\varepsilon_{\diamond}) = \dfrac{\big(\Tr[O]^2+\Tr[O^2]\big)}{2^n(2^{n}+1)} + \varepsilon_{\diamond}^{R}\norm{O}_{\infty} \,.
    \ee
    {We remark that this concentration is independent of the choice of encoding.}
\end{theorem}
Compared to Theorem~\ref{expressibility_thm1}, this implies a stronger concentration given by the inverse exponential dependence on the total space $d=2^n$. Crucially, Theorem~\ref{expressibility_thm2} tells us that if the ensemble $\Ubb_{R}$ forms a $2$-design, the probability of picking a $U_R$ for which the expectation value (for any given input) differs from $\mu$ by more than $\delta$ is exponentially suppressed. Hence, we will need exponentially many shots to recognise the observable from $\mu$. In Fig.~\ref{fig:reservoir_expressibility_numerics}, numerics support this result by showing exponential concentration when the unitaries are deep enough. Also in the case of the reservoir, unitaries similar to the one in Eq.~\eqref{numerics_ansatz} are considered, with randomly chosen parameters instead of data inputs.

While these two theorems are given in terms of some classical input data, in general they can be extended to quantum data. Suppose that the goal is to learn some unknown quantum process~\cite{innocenti2023potential,Ghosh_PRL_CQQ_2019,Ghosh_NPJQI_QQC_2019,suprano2023experimental}, and that to do so we are given a set of states $\{\rho_{k}\}_{k=1}^{D}$, along with some set of expectation values associated to each state $\{\vec{y}^{(k)}\}_{k=1}^{D}$. We then perform a QELM learning task by letting the input states evolve with the hidden space through some unitary $U_{R}$, and perform measurements which we use as inputs of a simple linear regression. While in this case the input data cannot be associated with classical data, it is normally always true that  these states must have undergone some unitary evolution drawn from a unitary ensemble. Then, even if we do not know the unitary ensemble, Theorem \ref{expressibility_thm1} holds as long as the underlying unitary ensemble approaches a $2$-design. Furthermore, Theorem \ref{expressibility_thm2} made no assumptions regarding the input data, which implies that its validity easily extends also in this scenario.
\begin{figure}
    \centering
    \includegraphics[scale=0.6]{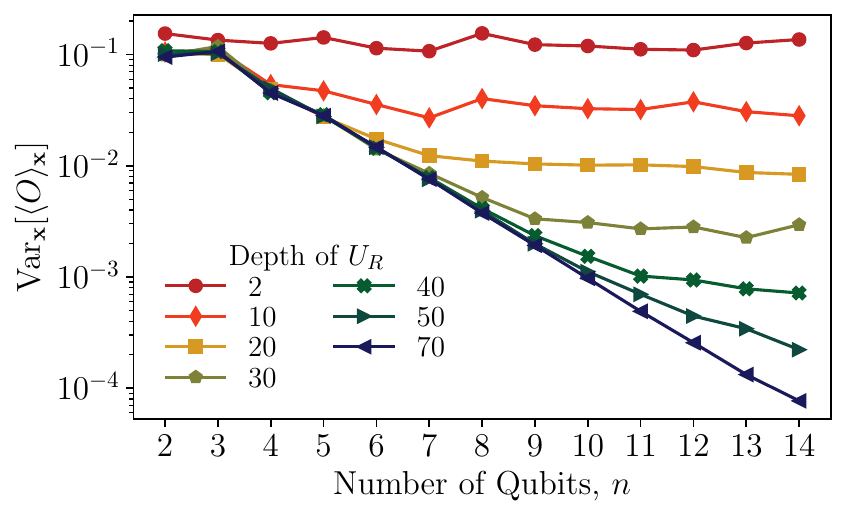}
    \caption{\textbf{Reservoir Haar-expressivity-induced concentration.} Variance of the observable $Z_{1}Z_{2}$ over a set of inputs uniformly sampled from $[-\pi,\pi]$, as a function of the number of total qubits $n=n_{a}+n_{h}$ and for different depths of the reservoir unitary similar to the definition in Eq.~\eqref{numerics_ansatz}.}
    \label{fig:reservoir_expressibility_numerics}
\end{figure}
\subsection{Entanglement-induced concentration}\label{sec:entanglement_concentration}
In this section we explore entanglement as another source of exponential concentration and extend the results from QNNs~\cite{marrero2020entanglement} to QELMs. In particular, we find that if the subspace onto which the observable acts non-trivially is highly entangled to the rest of the system, then we expect a deterministic concentration when computing the expectation value. Intuitively, tracing out qubits of highly entangled states leads to reduced states which are close to maximally mixed. Let us formalize this in the following Theorem.
\begin{theorem}[Entanglement-induced concentration]\label{entanglement_thm}
    Suppose an observable that acts non-trivially on a subspace $\HC_{k}$ of the entire Hilbert space $\HC$, so that $O = O_{k} \otimes \IC_{\bar{k}}$. Then, the concentration of its expectation value around the input-independent point $\mu = \Tr[O]/2^n$ will be bounded by
    \be\label{Eq:Conc_Ent}
    \left| \expval{O}_{\xv} - \mu\right| \leq \norm{O_k}_\infty\sqrt{2\ln{2}}S\biggr(\rhot_k(\xv)\biggr\lVert \dfrac{\IC_k}{2^k}\biggr)^{1/2} \,,
    \ee
    where $S(\cdot\lVert\cdot)$ is the relative entropy and $\rhot_k(\xv) = \Tr_{\bar{k}}(\rhot(\xv))$ represents the final reduced state on subspace $\HC_k$.
\end{theorem}
A proof of Theorem \ref{entanglement_thm} is provided in Appendix \ref{appendix_exp_concentration}. Crucially, for states that obey a volume-law scaling~\cite{bianchi2022volume}, i.e. $S(\rhot_k(\xv)\lVert \IC_k/2^{k}) \in \OC(2^{-n})$, then the difference between the observable and $\Tr[O]/2^n$ will always be exponentially small in the number of qubits. On the other hand, for states that obey an area-law scaling~\cite{eisert2010colloquium}, i.e. $S(\rhot_k(\xv)\lVert \IC_k/2^{k}) \in \OC(1)$, the bound in Theorem \ref{entanglement_thm} becomes loose and exponential concentration may be avoided. However, we remind the reader that even in the case where we have area-law scaling, this does not preclude concentration. Other sources of exponential concentration might still play a role and therefore they ought to be considered.

\subsection{Global-measurement-induced concentration}\label{sec:global_concentration}
We now consider the effect of global measurements. Indeed, previous work~\cite{cerezo2020cost, thanasilp2022exponential} has already shown that a global observable, which acts non-trivially on all $n$ qubits, can lead to an exponential concentration of its expectation value. In what follows, to distinguish this phenomenon from Haar-expressivity-induced and entanglement-induced concentration, we consider a separable initial state as well as encoding and reservoir unitaries in a tensor product form. This setting ensures no entanglement and low Haar-expressivity for both unitaries. By considering a projective measurement, we formalize this concept in the following theorem:
\begin{theorem}[Global measurement-induced concentration]\label{global_meas_thm}
Suppose an observable $O = \ketbra{m}{m}$, i.e. a projective measurement onto state $\ket{m}=\ket{m_1\dots m_n}$. Consider an initial separable state $\rho_0=\bigotimes_{k=1}^n \rho_0^{(k)}$. Suppose that the encoding unitary creates no entanglement, so that: $U(\xv)=\bigotimes_{k=1}^{n_{a}} U_k(x_k)$ where $x_k$ is an input component of $\xv$, and all are uniformly sampled from $[-\pi, \pi]$. Similarly, assume furthermore that the reservoir has the form $U_{R} =\bigotimes_{k=1}^n V_k$. Then we have
\be\label{Eq:Conc_Gl_Meas}
\mathrm{Pr}_{\xv}[|\expval{O}_{\xv} - \Ebb_{\xv}[\expval{O}_{\xv}]\geq \delta] \leq \dfrac{\alpha\prod_{k=1}^{n_{a}} G_k(\varepsilon_{\Ubb_{x_k}})}{\delta^2}\,,
\ee
where $\varepsilon_{\Ubb_{x_k}}= \norm{\AC_{\Ubb_{x_k}}\biggr(\rho_0^{(k)^{\otimes 2}}\biggr)}_1$ is the Haar-expressivity measure of the local unitary $U_k(x_k)$ and $\alpha= \prod_{j=n_{a}+1}^{n}\left|\braket{0|V_j}{m_j}\right|^4$. The term $G_k(\varepsilon_{\Ubb_{x_k}})$ is given by

\be
G_k(\varepsilon_{\Ubb_{x_k}}) = \biggr(\frac{1}{3}+\varepsilon_{\Ubb_{x_k}}\biggr(\varepsilon_{\Ubb_{x_k}}+\sqrt{\frac{4}{3}}\biggr)\biggr)^{1/2}\,.
\ee
\end{theorem}
A proof of the Theorem is given in Appendix \ref{appendix_exp_concentration}. This result suggests that using global measurements, even in the case of simple encoding and/or reservoir unitaries, is not advised as it can lead to exponential concentration of the observable. In particular, Theorem \ref{global_meas_thm} points out that for global measurements, it is the Haar-expressivity of single-qubit unitaries that dictates the exponential concentration. As shown in the plot provided in Fig.~\ref{fig:global_meas_numerics}, low-depth single-qubit unitaries already cause exponential concentration. 

\begin{figure}
    \centering
    \includegraphics[scale=0.6]{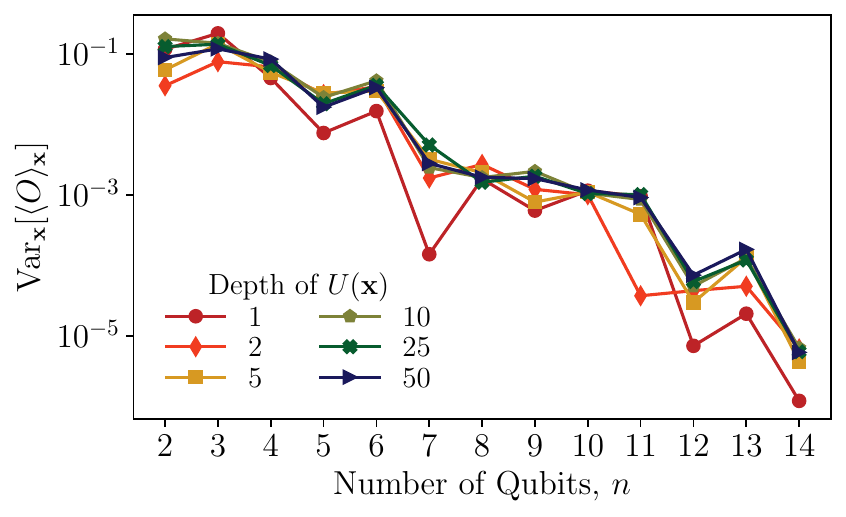}
    \caption{\textbf{Global measurement-induced concentration} Variance of a global observable as a function of the system size $n$, varying the depth of the encoding unitary. The variance is taken over an ensemble of inputs uniformly sampled from $[-\pi,\pi]$.}
    \label{fig:global_meas_numerics}
\end{figure}

\subsection{Noise-induced concentration}\label{sec:noise_concentration}
Noise is the last source of exponential concentration which can hinder the performance of QELM. As of today, {Noisy Intermediate-Scale Quantum (NISQ)} devices are in fact prone to making errors caused by different phenomena, be it decoherence, dephasing or other. These are one of the major obstacles yet to be overcome in order to achieve full control of a quantum device. On top of this, Ref.~\cite{wang2020noise} has already shown the existence of noise-induced {barren plateaus}. This is not limited to gradient estimation, but also cost estimation in general~\cite{thanasilp2022exponential}. In this section, we show that Pauli noise, a quite general noise model, leads to exponential concentration in a QELM model.

\begin{figure}[htbp]
    \centering
    \includegraphics[scale=0.6]{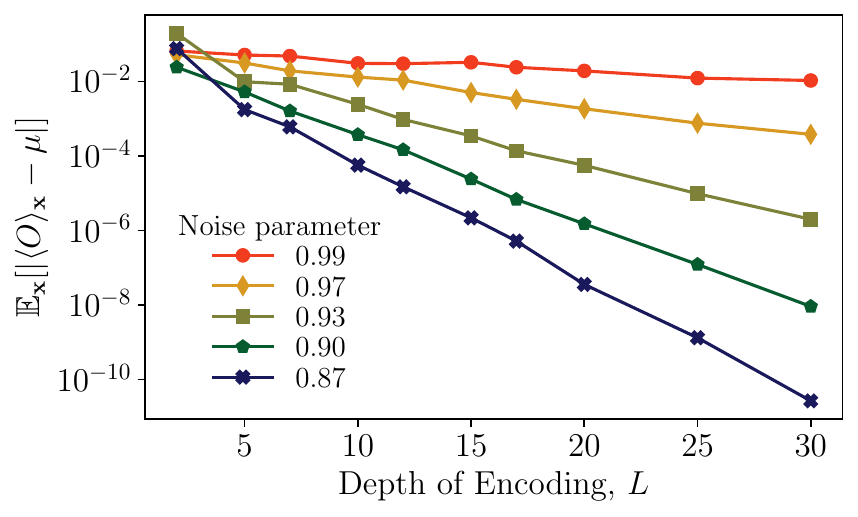}
    \caption{\textbf{Noise-induced concentration} Distance of $\expval{Z_1}_{\xv}$ from $\mu = 0$ as a function of the number of encoding layers, and for different noise parameters $p$. Here we consider $n_{a} = 7,\ n_{h}=1$. The distance is averaged over inputs uniformly sampled from $[-\pi,\pi]$. Exponential concentration is seen in all cases, with increasing slope as the noise in the system increases.}
    \label{fig:noisy_numerics}
\end{figure}
Consider a $L$-layered encoding, subject to Pauli noise. Then the embedded state will be
\be\label{noisy_encoding}
\rho(\xv) = \NC \circ \UC_{L}(\xv_L) \circ \NC \circ...\circ \NC \circ \UC(\xv_1)\circ \NC(\rho_0) \,,
\ee
where $\UC(\xv_i)(\cdot) = U(\xv_i)(\cdot)U(\xv_i)\ad$ and $\{\xv_i\}_{i=1}^L$ are some functions of $\xv$. Moreover, each noise term $\NC = \NC_1 \otimes...\otimes \NC_{n_{a}}$ acts on each single qubit in the following way
\be
\NC_j(\sigma)=q_\sigma \sigma,\hspace{1cm} \forall j= 1,...,n\hspace{1cm}\sigma \in\{X,Y,Z\} \,.
\ee
Finally let us define the characteristic noise parameter $q= \max \{|q_X|, |q_Y|, |q_Z|\}$. Then, the following theorem holds true.
\begin{theorem}[Noise-induced concentration]\label{noisy_encoding_thm}
    Consider the $L$-layered encoding as defined in  Eq.~\eqref{noisy_encoding} with $q < 1$. Then, the concentration around an input-independent point of the expectation value of an observable as defined in Eq.~\eqref{obs_def} can be bounded as
    \be\label{eq:noisy_concentration}
    \biggr|\expval{O}_{\xv} - \dfrac{\Tr[\widetilde{O}_{\Lambda}]}{2^{n_{a}}}\biggr| \leq \norm{\Lambda(O)}_{\infty}\biggr(\dfrac{1}{b}q^{b(L+1)}S_2\biggr(\rho_0\biggr\lVert\frac{\IC}{2^{n_{a}}}\biggr)\biggr)^{\frac{1}{2}} \,,
    \ee
    where $\widetilde{O}_{\Lambda} = \Tr[\Lambda(O)(\IC\otimes \ketbra{0}{0})]$, $b = 1/(2\ln{2})$ and $S_2(\cdot\lVert\cdot)$ denotes the sandwiched $2$-Rényi relative entropy.
\end{theorem}
The inequality in Eq.~\eqref{eq:noisy_concentration} points out that if $L \in \OC(n_{a})$, then the deviation of the observable from an input-independent value will be exponentially vanishing with respect to the dimension of the accessible space. Numerics in Fig.~\ref{fig:noisy_numerics} demonstrate the characteristic exponential dependence on the depth of the encoding for different noise parameters.
\section{Discussion}
We have conducted an analytical and numerical study on the interplay between the classical data encoding and exponential concentration phenomena on the expressivity and prediction capabilities of QELMs. The quantum substrate of a QELM framework here consists of three components: a classical data encoder, a quantum reservoir, and a measurement readout. The unitary encoding strategy determines achievable Fourier frequencies. The quantum reservoir acts as a feature map induced by many-body quantum dynamics. 

The measurement readout consists of several observables, each of which selects a function from the class of classical mappings defined by the reservoir. The net result of the reservoir and measurement steps is a set of basis functions. The model is then trained by fitting a linear combination of these functions to the training data.

Based on this framework, we started by discussing the performance of a QELM from two perspectives: its Fourier-expressivity and whether the model could be replaced by a classical surrogate. In particular, we found that the Fourier-expressivity of a QELM is upper bounded by the number of frequencies resulting from the encoding, the number of observables, and the globality of observables. Hence, the encoding and measurement jointly determine the Fourier-expressivity. We recall that a QELM model which yields a Fourier-concentrated prediction can be efficiently simulated on a classical computer~\cite{schreiber2022classical,landman2022classically,sweke2023potential} and so using an encoding strategy that results in exponentially many achievable frequencies is a necessary condition to achieve quantum advantage. On this basis we discourage the use of Pauli encodings.

Our study highlights that the reservoir dynamics has a significant effect on the Fourier spectra of a QELM model. We compared the Fourier spectra resulting from four reservoir settings: no reservoir, integrable Ising reservoir, chaotic Ising reservoir, and Haar random reservoir and defined the richness of Fourier modes as the proportion of non-zero Fourier coefficients averaged over Pauli strings. We found that the richness decreases exponentially with the system size for the first two reservoir settings, and saturates for the last two reservoirs. As adding exponentially many Fourier modes is classically inefficient, selecting a quantum reservoir with rich dynamics may lead to QELMs that can outperform classical surrogates. 

However, there is a balance to be struck in this regard. While chaotic Ising and Haar random reservoirs (i.e. scrambling reservoirs) lead to a rich set of basis functions, such reservoirs are also prone to exponential concentration. This washes out the effect of the input data, hindering the model's generalization power and thus limiting the scalability of such QELMs. Crucially, this implies that there exists a trade-off between Fourier-expressivity and generalizability: The highly expressive encoding unitaries and global measurements enhance the Fourier-expressivity of the prediction, but lead to exponential concentration, making the model insensitive to input data. More generally, we identified four sources of exponential concentration: the Haar-expressivity of both encoding and reservoir unitaries, entanglement, global measurement, and noise.

Our work also provides insights for understanding the limitations and advantages of using quantum computational substrate in the context of (classical) ELMs. Compared with a back-propagation trained feed-forward network, one major disadvantage of ELMs is the need for a very large fan-out to achieve a comparable performance. For instance, to learn the MNIST-database of handwritten digits with dimension $784$, it requires a fan-out of 20 and 15680 hidden neurons ~\cite{de2015comparison}. The number of hidden neurons somewhat corresponds to the number of observables in QELMs, which cannot be efficiently increased by using a quantum computational substrate. Hence, for such tasks the number of {measurement shots} needs to be of the same magnitude as the number of data features and this limitation may still restrict the performance of QELMs. Nevertheless, in this work we showed that the quantum computational substrate can provide additional activation functions which are potentially inefficient for a classical computer, in particular Fourier series consisting of exponentially many modes. Consequently, QELMs create opportunities to solve learning tasks, for which common activation functions are not suited.

We remark that the role of hidden qubits has not been discussed in this work. More specifically, the Fourier-expressivity is independent of number of hidden qubits $n_{h}$. Further work is hence required to ascertain the effect of $n_{h}$ on the QELM's prediction. In particular, it would be valuable to explore if a larger hidden space would lead to some general improvements of QELM, or it would just amount to a different inductive bias, such that $n_{h}$ should be purely viewed as a hyperparameter. It would also be interesting to compare performance of other reservoir types and of reservoirs with different $n_{h}$ on a given task and investigate a strategy of optimizing $n_{h}$.

Given the close relationship between QELMs and Quantum Reservoir Computing (QRC), which employs a quantum reservoir for time-series prediction, our analysis can be extended to examine the expressivity and the exponential concentration in QRC. This understanding is key to designing QRC that forecast dynamical systems based on classical data encoding and measurement readout \cite{fujii2017harnessing, nakajima_PRApp_spatialmultiplexing_2019, Kutvonen2020optimizing,kubota2022quantum,fry2023optimizing, wudarski2023hybrid}. As the measurement readout can introduce issues of exponential concentration, the QRC prediction on unseen time-series data might also become input-agnostic when the size of a quantum reservoir grows large. Our results indicate that adopting QRC that outputs quantum data \cite{Ghosh_NPJQI_QQC_2019, Ghosh_PRL_CQQ_2019}, or its variant such as Quantum Next Generation Reservoir Computing that forecasts future quantum states \cite{sornsaeng2023quantum} {and the QRC with continuous input-series~\cite{senanian2024microwave}}, may circumvent the scalability issues imposed by exponential concentration by working directly with structured quantum data. However, we leave the detailed analysis of quantum reservoir systems for time-series prediction for future work.

\section{Acknowledgements}
TC acknowledges valuable discussions with Rodrigo Martínez-Peña and the funding support from the NSRF via the Program Management Unit for Human Resources \& Institutional Development, Research and Innovation [grant number B39G670016], as well as from Thailand Science Research and Innovation Fund Chulalongkorn University (IND66230005). ST and ZH acknowledge support from the Sandoz Family Foundation-Monique de Meuron program for Academic Promotion. WX acknowledges support from the NCCR SPIN, funded by the Swiss National Science Foundation.

\bibliography{quantum.bib,otherbib.bib}

\clearpage
\onecolumngrid
\setcounter{theorem}{0}
\setcounter{proposition}{0}
\setcounter{corollary}{0}

\appendix
\vspace{0.5in}
\begin{center}
	{\Large \bf Appendix} 
\end{center}

\section{Fourier decomposition analysis}\label{AppFourierDecomp}

\subsection{Fourier decomposition for vector input} \label{AppMultiFourier}
In this appendix, we generalize the case of scalar input addressed in the main text and show that the prediction of a QELM with vector input $\xv\in\mathbb{R}^{d_x}$ can be expressed as a multivariate Fourier series. Let the initial state of the accessible qubits before encoding be $\rho_0 = \sum_{i,j} \alpha_{ij}\ket{i}\bra{j}$. We consider an encoding unitary of the form\\
\begin{equation}
    U(\xv)=e^{iH_1 x_1}\otimes{}e^{iH_2 x_2}\otimes\cdots\otimes{}e^{iH_{d_x}x_{d_x}}\,,
\end{equation}
where the generator $H_l$ of dimension $\zeta_l$ corresponding to the $l$-th component of $\xv$ has eigenvalues $\{\lambda^{(l)}_1, \lambda^{(l)}_2, \dots, \lambda^{(l)}_{\zeta_l}\}$ for $l=1, 2,\dots, d_x$.
Then, the state of accessible qubits after encoding is\\
\begin{align}
     \rho(\xv) &= U(\xv) \rho_0 U(\xv)^\dag\\
     &= \sum_{i,j} \left(\prod_{l=1}^{d_x} e^{(\lambda^{(l)}_{\nu{}(i;l)}-\lambda^{(l)}_{\nu{}(j;l)})x_l}\right) \alpha_{ij}\ket{i}\bra{j}\\
     &= \sum_{i,j} e^{i(\bm{\lambda}_{i}-\bm{\lambda}_{j})\cdot\xv} \alpha_{ij}\ket{i}\bra{j}\,.
\end{align}
In the second line, the function $\nu{}(i;l)$ gives the eigenstate index of the subspace, within which the input component $x_l$ is encoded, when the qudit representing all accessible qubits is in state $\ket{i}$. In the third line, the vector of eigenvalues is defined as $\bm{\lambda}_i:=\left(\lambda^{(1)}_{\nu(i;1)},\ \lambda^{(2)}_{\nu(i;2)},\cdots,\ \lambda^{(d_x)}_{\nu(i;d_x)}\right)^\intercal$.

Compared with Eq.~\eqref{Eq:StateAfterEnc1d} for scalar input cases, we observe that, the state after encoding a vector input only differs in the Fourier basis. Instead of $e^{i(\lambda_{j}-\lambda_{i})x}$ for scalar input data, we have here $e^{i(\bm{\lambda}_{j}-\bm{\lambda}_{i})\cdot\xv}$. Analogously, by Eq.~\eqref{Eq:ReservoirEvo} and \eqref{obs_def}, we obtain the expectation value of readout observable
\begin{equation}
    \expval{O}_{x}= \sum_{\bm{\omega}\in\Omega} a_{\bm{\omega}}e^{i\bm{\omega}\cdot\xv}\,,
\end{equation}
where the set of vectorial Fourier frequencies are given by
\begin{equation}
    \Omega=\{\bm{\lambda}_j-\bm{\lambda}_i: i,j = 1, 2,\dots ,2^{n_{a}}\}\,, 
\end{equation}
and the corresponding Fourier coefficients are
\begin{equation}
    a_{\bm{\omega}}=\sum_{i, j \mathrm{\,s. t.\,} \bm{\lambda}_j-\bm{\lambda}_i =\bm{\omega}} \alpha_{ij}\bra{j,0}U_R^\dag{}O U_R \ket{i,0}\,.
\end{equation}
Consequently, by Eq.~\eqref{Eq:sysout}, the prediction of a QELM with a vector input can be expressed as a multivariate Fourier series
\begin{align}
    f_{\bm{\eta}}
    &=\sum_{\bm{\omega}\in\Omega} b_{\bm{\w}} e^{i\bm{\omega}\cdot \xv}\,,
\end{align}
where $b_{\bm{\w}}=\sum_{k=1}^M \eta_k  a_{\bm{\w}}^{(k)}$.

\subsection{Upper bound of Fourier-expressivity}\label{App:UpperBound}
In this Appendix, we rigorously prove the upper bound of Fourier-expressivity given by Theorem~\ref{Observa1} in main text.
\begin{theorem}[Upper bound of QELM's Fourier-expressivity]
Consider a QELM as defined above with model prediction $f_{\bm{\eta}}$. Let $M$ be the number of observables, $\Omega$ be the set of achievable frequencies, and $n_{o}$ be the number of measured qubits, then
    \begin{align}
        \mathcal{F}\left[f_{\bm{\eta}}\right] &\leq\operatorname{min}\{M, |\Omega|, 4^{n_{o}}\} \,,
    \end{align}
where $|\Omega|\leq 4^{n_{a}}$.
\end{theorem}
\begin{proof}
We recall (Eq.~\eqref{Eq:MOutcome} and \eqref{Eq:AutoF}) the measurement outcome can be written as 
\begin{equation}\label{A10}
    \expval{O_{k}}_x = \sum_{\omega\in\Omega} a^{(k)}_{\omega}e^{i\omega{x}}\,
\end{equation}
and the prediction of the QELM is
\begin{align}
    f_{\bm{\eta}}&=\sum_{k=1}^M \eta_k \expval{O_{k}}_x\label{A11}\\
    &=\sum_{\omega\in\Omega} b_{\w} e^{i\omega x}\,.
\end{align} 
As both consist of a linear combination of terms, it is convenient to use a matrix representation. Recall that we write the vector of trainable weights as $\bm{\eta}=[\eta_1, \cdots \eta_M]^\intercal$ and the prediction's Fourier coefficients as $\bm{b}=[b_{\omega_1}, \cdots b_{\omega_{|\Omega|}}]^\intercal$. We further define $\bm{A}\in\mathbb{C}^{|\Omega|\cross{}M}$ to be the matrix consisting of the Fourier coefficients of each of the $M$ observables as its columns, i.e.,
\be
    \bm{A} = 
    \begin{pmatrix}
        a_{\w_1}^{(1)} & a_{\w_1}^{(2)} & \cdots & a_{\w_1}^{(M)} \\
        a_{\w_2}^{(1)} & a_{\w_2}^{(2)} & \cdots & a_{\w_2}^{(M)} \\
        \vdots & \vdots & \ddots & \vdots \\
        a_{\w_{|\Omega|}}^{(1)} & a_{\w_{|\Omega|}}^{(2)} & \cdots & a_{\w_{|\Omega|}}^{(M)}
    \end{pmatrix} \,.
\ee
From Eq.~\eqref{A10} and \eqref{A11}, we can write
\be\label{Eq:AB}
    \bm{b}=\bm{A}\cdot \bm{\eta}  \,.
\ee
We observe from Eq.~\eqref{Eq:AutoF} that the prediction is a linear combination of Fourier single mode functions which are orthogonal. Hence, the Fourier-expressivity $\mathcal{F}$ is exactly the degrees of freedom of the coefficients $\{b_{\w_1},\cdots b_{\w_{|\Omega|}}\}$, which are components of $\bm{b}$. Moreover, in Eq.~\eqref{Eq:AB} each component of $\bm{b}$ is a linear combination of column vectors of matrix $\bm{A}$ with trainable weights $\vec{\eta}$. Hence, the degrees of freedom of $\bm{b}$'s components equals the rank of $\bm{A}$ and we obtain
\be\label{Eq:F=rank}
    \mathcal{F}\left[f_{\bm{\eta}}\right]=\operatorname{rank}(A) \, .
\ee

Since the elements of $\bm{A}$ are controlled by both the reservoir unitary and the measurement, we next decompose $\bm{A}$ into a product of two matrices, which represent the dynamics of reservoir and the measurement, respectively. Recall that any Hermitian observable can be expanded in Pauli basis. We assume $n_{o}$ qubits are measured ($n_{o}\leq n_{a}+n_{h}$), and let $\{{P_1}, {P_2}, \dots, {P}_{4^{n_{o}}}\}$ denote the Pauli basis,
\begin{equation}
    \Bigl\{{P_r}=\left(\bigotimes_{i=1}^{n_{o}}\sigma_i\right)\otimes\IC: \sigma_i\in\{\sigma_x,\sigma_y,\sigma_z,I\}\Bigr\} \,,
\end{equation}
whose elements form a complete operator basis. Therefore, any Hermitian observable can be written as a linear combination of these Pauli strings, i.e.
\begin{equation}\label{Eq:ICDecompose}
    O_{k}=\sum_{r=1}^{4^{n_{o}}}\gamma_r^{(k)}{P_r} \,,
\end{equation}
where $\gamma_r^{(k)}\in\mathbb{R}$. We denote the Fourier coefficients associated to ${P_r}$ as 
\begin{equation}
    p^{(r)}_{\omega}:=\sum_{i, j \mathrm{\,s. t.\,} \lambda_j - \lambda_i =\omega} \alpha_{ij}\bra{j,0}U_R^\dag{}{P_r} U_R \ket{i,0}\,.
\end{equation}
From Eq.~\eqref{Eq:FParam} we observe that the Fourier coefficients $a_{\omega}^{(k)}$ are linear in $O_{k}$. Thus, by substituting Eq.~\eqref{Eq:ICDecompose} into Eq.~\eqref{Eq:FParam}, we obtain
\begin{equation}
    a_{\omega}^{(k)}=\sum_{r=1}^{4^{n_{o}}}\gamma_r^{(k)} p^{(r)}_{\omega}\,.
\end{equation}\\
We now write this equation into matrix form again
\begin{equation}\label{Eq:APG}
    \bm{A}=\bm{P}\,\bm{\Gamma},
\end{equation}
where matrices $\bm{P}\in \mathbb{R}^{|\Omega|\cross 4^{n_{o}}}$ and $\bm{\Gamma}\in \mathbb{R}^{4^{n_{o}}\cross M} $ are respectively defined as
\begin{equation}
    \bm{P} := 
    \begin{pmatrix}
        p^{(1)}_{\omega_1} & p^{(2)}_{\omega_1} & \cdots & p^{(4^{n_{o}})}_{\omega_1} \\
        p^{(1)}_{\omega_2} & p^{(2)}_{\omega_2} & \cdots & p^{(4^{n_{o}})}_{\omega_2} \\
        \vdots & \vdots & \ddots & \vdots \\
        p^{(1)}_{\omega_{|\Omega|}} & p^{(2)}_{\omega_{|\Omega|}} & \cdots & p^{(4^{n_{o}})}_{\omega_{|\Omega|}}
    \end{pmatrix}
     \,
\end{equation}
and
\begin{equation}
    \bm{\Gamma} = 
    \begin{pmatrix}
        \gamma_{1}^{(1)} & \gamma_{1}^{(2)} & \cdots & \gamma_{1}^{(M)} \\
        \gamma_{2}^{(1)} & \gamma_{2}^{(2)} & \cdots & \gamma_{2}^{(M)} \\
        \vdots & \vdots & \ddots & \vdots \\
        \gamma_{4^{n_{o}}}^{(1)} & \gamma_{4^{n_{o}}}^{(2)} & \cdots & \gamma_{4^{n_{o}}}^{(M)}
    \end{pmatrix}
    \,.
\end{equation}
Finally, from Eq.~\eqref{Eq:F=rank} and \eqref{Eq:APG}, we obtain the following upper bound of the Fourier-expressivity in terms of the rank of matrix $\bm{A}$.
\begin{align}
        \mathcal{F}\left[f_{\bm{\eta}}\right]&=\operatorname{rank}(\bm{A})\\
        &\leqslant \operatorname{min}\{\operatorname{rank}(\bm{P}),\operatorname{rank}(\bm{\Gamma})\}\\
        &= \operatorname{min}\{M, |\Omega|, 4^{n_{o}}\} \,.
\end{align}
We remark that $|\Omega|<4^{n_{a}}$ since, by Eq.~\eqref{Eq:FreqSet}, the achievable frequencies are differences of eigenvalues. For $n_{a}$ accessible qubits, the encoding unitary has $2^{n_{a}}$ eigenvalues which provide at most $4^{n_{a}}$ distinct differences of the eigenvalues.

\end{proof}

\section{QELM with Haar random reservoirs}\label{stats}

In this appendix we look into details of Fourier coefficients and analytically show that there are QELMs that can not be surrogated via RFF method. In particular, we show that Haar reservoirs lead to anti-concentrated Fourier spectrum which contradicts a necessary condition to perform RFF. Despite the model being not RFF surrogatable, the QELM with Haar random reservoirs suffer from exponential concentration. As a consequence, we further show that there exists a trivial surrogate model with data independent output.

\subsection{Moments of Fourier Coefficients}

\subsubsection{Statistics of Fourier summands}
First, we analyse the statistics of the Fourier coefficients of a QELM with a Haar random reservoir. We assume Pauli string observables in the composite space of all the qubits. We recall that they satisfy the following properties
\begin{equation}
    \Tr[O] = 0\ \operatorname{and}\ \Tr[O^2] = 2^n \;,
\end{equation}
where $O$ is one of the Pauli string from the observable set. By computing mean, variance and covariance of the Fourier coefficients, we obtain the following Proposition.
\begin{proposition}[Statistics of Fourier Summands]\label{Prop:StatFS}
For all summand of Fourier coefficients defined in Eq.~\eqref{Eq:FParam}, denoted as $a_{uv}:=\alpha_{uv}\bra{u,0}U_R^\dag{}O U_R \ket{v,0}$ where $O$ is a Pauli string, we have\\
\begin{equation}
    \mathbb{E}_{U_R}[a_{uv}] = 0 \;,
\end{equation}
and
\begin{equation} 
 \Var_{U_R}(a_{uv}) =   
\begin{dcases}\label{Eq:B3}
    |\alpha_{uv}|^2\dfrac{d}{d^2-1} &\quad\operatorname{for}\ u\neq v \\
    |\alpha_{uv}|^2\dfrac{1}{d+1} &\quad\operatorname{for}\  u= v \;,\\
\end{dcases}  
\end{equation}
where for both expectation values the reservoirs are sampled from \textbf{Haar measure} and $d=2^n$ corresponds to the system size. Moreover, the covariance of summands is given by
\begin{equation}\label{Eq:B4}
    \Cov_{U_R}(a_{uv}, a_{jk})=
    \begin{dcases}
    -\alpha_{uv}^* \alpha_{jk}\dfrac{1}{d^2-1} &\quad\operatorname{for}\ u = v \operatorname{and} j=k\\
    0 &\quad\operatorname{for\ all\ distinct}\ u,v,j,k\;. 
    \end{dcases}
\end{equation}
\end{proposition}
\begin{proof}
We start with the expectation value of any summand, by its linearity we obtain
\begin{align}
    \mathbb{E}_{U_R}[a_{uv}]&=\alpha_{uv}\int_{\mathbb{U}(d)}d\mu(U)\bra{u,0} U_R^\dag O U_R\ket{v,0}\\
    &=\alpha_{uv}\bra{u,0}\left(\int_{\mathbb{U}(d)}d\mu(U) U_R^\dag O U_R\right)\ket{v,0}\\ 
    &= \alpha_{uv}\bra{u,0} \frac{\Tr[O]}{d}\IC \ket{v,0}\\
    &= \alpha_{uv}\frac{\Tr[O]}{d} \delta_{uv}\,,
\end{align}
where in the third line we use standard Haar integral identities~\cite{cerezo2021cost}. Recall that Pauli strings are traceless, and so
\begin{equation}
    \mathbb{E}_{U_R}[a_{uv}] = 0\,.
\end{equation}
Next, we analyse the correlation among summands. By the properties of the Haar integral, we obtain
\begin{align}\label{coeffs cov}
    \mathbb{E}_{U_R}[a^*_{uv}a_{jk}] &=\alpha^*_{uv}\alpha_{jk} \bra{v,0}\otimes\bra{j,0} \int_{\mathbb{U}(d)}d\mu(U) U_R^{\dag\otimes2} O^{\otimes2} U_R^{\otimes2} \ket{u,0}\otimes\ket{k,0}\\
    & = \alpha^*_{uv}\alpha_{jk}\frac{1}{d^2-1}\left(\bra{v,0}\otimes\bra{j,0}\left( \left(\Tr^2[O]-\frac{1}{d} \Tr[O^2]\right)\IC + \left(\Tr[O^2]-\frac{1}{d} \Tr^2[O])\right)\text{SWAP}\right) \ket{u,0}\otimes\ket{k,0} \right)\\
    & = \alpha^*_{uv}\alpha_{jk}\frac{1}{d^2-1}\left(\left(\Tr^2[O]-\frac{1}{d} \Tr[O^2]\right)\delta_{uv} \delta_{jk} + \left(\Tr[O^2]-\frac{1}{d} \Tr^2[O]\right)\delta_{uj}\delta_{vk}\right) \;,
\end{align}
where $\text{SWAP}$ is a swap operator between two copies.

Let $u=j$ and $v=k$, we obtain: 
\begin{equation}
\Var_{U_R}(a_{uv}) = \mathbb{E}_{U_R}\left[\abs{a_{uv}}^2\right] =    
\begin{dcases}
    |\alpha_{uv}|^2\frac{1}{d^2-1}\left(\Tr[O^2]-\frac{1}{d}\Tr^2[O]\right) &\quad\operatorname{for}\ u\neq v\\
   |\alpha_{uv}|^2\frac{1}{d(d+1)}\biggr(\Tr[O^2]+\Tr^2[O]\biggr) &\quad\operatorname{for}\ u= v\\
\end{dcases}
\end{equation}
By the properties of Pauli strings $\Tr[O^2]=d$ and $\Tr[O]=0$, we obtain
\begin{equation}
  \Var_{U_R}(a_{uv}) =   
\begin{dcases}
    |\alpha_{uv}|^2 \frac{d}{d^2-1} &\quad\operatorname{for}\  u\neq v\\
    |\alpha_{uv}|^2 \frac{1}{d+1} &\quad\operatorname{for}\  u= v\\
\end{dcases}  
\end{equation}
Finally, by Eq.~\eqref{coeffs cov} we also obtain the desired properties of covariance of the summands
\begin{equation} \label{eq:cov}
    \Cov_{U_R}(a_{uv}, a_{jk})=
    \begin{dcases}
    -\alpha_{uv}^* \alpha_{jk}\dfrac{1}{d^2-1} &\quad\operatorname{for}\ u = v \operatorname{and}\ j=k \ \operatorname{and}\ u \neq j\\
    0 &\quad\operatorname{for\ distinct}\ u,v,j,k 
    \end{dcases}
\end{equation}
Note that two cases ``$u=v=j=k$'' and ``$u=j,\ v=k,\ \operatorname{and}\ u\neq v$'' are not considered in Eq. \eqref{eq:cov}. However, in both cases this quantity become variances of the summands, which are obtained in Eq. \eqref{Eq:B3}. 
\end{proof}
\subsubsection{Degeneracy of Fourier coefficients}
Recall that different pairs of eigen states might provide the same difference of their corresponding eigen values
\be\label{Eq:FCSummandAgain}
a_\omega = \sum_{u,v \text{ s.t. } \lambda_u- \lambda_v=\omega } a_{uv} \,.
\ee
Hence, the number of potential summands is relevant for the statistics of Fourier coefficients. This motivate us to introduce the degeneracy of Fourier frequencies:\\
\begin{definition}
    For an given encoding strategy we define the degeneracy of frequency $\omega$ as the number summands that contribute to frequency $\omega$,
\be 
g_\omega  \coloneqq \sum_{i,j=1}^{|\Omega|} \mathbb{1}\{ \lambda_i - \lambda_j = \omega\}  
\ee 
where $\mathbb{1}$ is the indicator function. The degeneracies for all frequencies are denoted as the degeneracy vector $\textbf{g}\in\mathbb{R}^{|\Omega|}$.
\end{definition}
\begin{lemma}[Degeneracy of Pauli re-uploading scheme]\label{prop:2}
    Considering a Pauli re-uploading scheme on $L$ qubits we have
    \be 
    g^{\text{Pauli}}_\omega = \binom{2L}{L-\omega}
    \ee
\end{lemma}
\begin{proof}
    By Eq. \eqref{Eq:Omega} we know that the set of achievable frequencies for Pauli encoding can be expressed as
    \begin{align} 
        \Omega = &\{ (\lambda_1^{(1)} - \lambda_2^{(1)}) + \cdots + (\lambda_1^{(L)} - \lambda_2^{(L)}) : \lambda_{1}^{(i)}, \lambda_{2}^{(i)} = \pm\frac{1}{2}\}  \nonumber \\ 
        = &  \{ (\lambda_1^{(1)} + \lambda_2^{(1)}) + \cdots + (\lambda_1^{(L)} + \lambda_2^{(L)}) : \lambda_{1}^{(i)}, \lambda_{2}^{(i)} = \pm\frac{1}{2} \} \nonumber \\
        = & \{ \lambda_1 + \cdots + \lambda_{2L}: \lambda_{j} = \pm\frac{1}{2} \} 
    \end{align}
    In the second step we use the symmetry in possible values of $\lambda_{j}^{i}$ and the third step is just a change of notation.

    Next, we choose $2L$ copies of $\lambda$ from $\{-1/2, +1/2\}$. We let $L+\omega$ of these are $+1/2$ and the rest are $-1/2$ and obtain the frequency $\frac{1}{2}(L+\omega)-\frac{1}{2}(L-\omega)= \omega$. Thus the number of combinations that correspond to frequency $\omega$ is $\binom{2L}{L-\omega}$. 
\end{proof}
\begin{lemma}[Degeneracy of exponential encoding] \label{prop:3}
      Let $\textbf{g}^{\text{Exp}}{(L)} \in \mathbb{N}_0^{3^L}$ be the degeneracy vector for all frequencies of an $L$ qubit exponential encoding with components $\textbf{g}^{\text{Exp}}_i{(L)} = g_{i-\frac{3^L+1}{2}}$. Then, vector $\textbf{g}^{\text{Exp}}{(L+1)}$ containing the degeneracies of $(L+1)$ qubit exponential encoding is given by: 
      \be 
      \textbf{g}^{\text{Exp}}{(L+1)} = [\textbf{g}^{\text{Exp}}{(L)}, \ 2\cdot \textbf{g}^{\text{Exp}}{(L)}, \ \textbf{g}^{\text{Exp}}{(L)}]
      \ee 
      where $[...]$ denotes concatenation. 
\end{lemma}
 \begin{proof}
Again recalling from Eq. \eqref{Eq:Omega} we have: 
\begin{align}
    \Omega_{L+1} = &\{ (\lambda_1^{(1)} - \lambda_2^{(1)}) + \cdots + (\lambda_1^{(L+1)} - \lambda_2^{(L+1)}) : \lambda_{1}^{(i)}, \lambda_{2}^{(i)} = \pm \frac{3^{i-1}}{2}\}  \nonumber \\ 
    = & \Omega_L + \{\lambda_1^{(L+1)} - \lambda_2^{(L+1)} : \lambda_{1}^{(L+1)}, \lambda_{2}^{(L+1)} = \pm \frac{3^{L}}{2}\} \nonumber \\ 
    = &  \{-\frac{3^L-1}{2},-\frac{3^L-1}{2}+1, \cdots , \frac{3^L-1}{2} \}   + \{-3^L, 0, 3^L\} \nonumber \\ 
    = & S_1 \cup S_2 \cup S_3 
\end{align}
where in the third line we use the frequency range of exponential encoding given by Eq. \eqref{OmegaExp} and $S_1 = \Omega_L \mathbb{+} \{-3^L\}$, $S_2 = \Omega_L + \{0\} = \Omega_L$, $S_3 = \Omega_L + \{3^L\}$. Here, the plus sign between sets denotes the Minkowski addition.

It is easy to see that $S_1,\ S_2,\ S_3$ are disjoint. If $\omega \in S_1$ the number of ways we can get $\omega$ is equal to the degeneracy of $\omega+3^L$ in $\Omega_L$ times the degeneracy of $3^L$ which is one. A similar argument can be made for $\omega \in S_3$.  For  $\omega \in S_2$ the degeneracy is equal to the the degeneracy of $\omega$ in $\Omega_L$ times the degeneracy of $0$ which is two. By arranging the vectors of degeneracies together in order we obtain $[\textbf{g}^{\text{Exp}}{(L)}, \ 2\cdot \textbf{g}^{\text{Exp}}{(L)}, \ \textbf{g}^{\text{Exp}}{(L)}]$. The ref. \cite{barthe2023gradients} discussed the degeneracies and distribution of the coefficients in detail and provided an alternative proof for degeneracies using convolutions.
\end{proof}
\begin{corollary}\label{col1}
     Within an exponential encoding scheme with L qubits the frequency 0 has the largest degeneracy $2^L$. Furthermore, sum of all degeneracies is $4^L$. 
 \end{corollary}
 \begin{proof}
A proof by induction is trivial using Lemma \ref{prop:3} and the fact that $\textbf{g}^{Exp}(1) = [1,2,1]$.
 \end{proof}
 \subsubsection{Statistics of Fourier coefficients}
 We further assume that  $\rho_0 = (\ketbra{+}{+})^{\otimes{}n_{a}}$ such that $\alpha_{ij} = \frac{1}{2^{n_{a}}} = \frac{1}{{d_{a}}} \ \forall i,j$. From now, we follow this convention. \\
 \begin{customthm}{7}[Variance of Fourier Coefficients]
 Given a QELM with Fourier coefficients $a_{\omega}$, their variance over reservoirs sampled from \textbf{Haar measure} is given by
 \be
\Var(a_\omega) = g_\omega \cdot \frac{d}{(d^2-1)d_{a}^2}\,,\ \ \text{for}\ \ \omega \neq 0 
 \ee
 where $g_\omega$ is the degeneracy of that frequency given in Lemma \ref{prop:2} or \ref{prop:3} depending on the chosen encoding scheme.\\ 
 And for $\omega = 0$ the variance is given by
 \be 
\Var(a_0^{\text{Exp}}) =\frac{1}{(d+1)d_{a}} - \frac{d_{a}-1}{(d^2-1)d_{a}} 
 \ee
and
\be 
\Var(a_0^{\text{Pauli}}) = d_{a} \cdot \frac{1}{d_{a}^2(d+1)} - (d_{a}^2-d_{a})\cdot \frac{1}{d_{a}^2(d^2-1)} +\biggl[\binom{2n_{a}}{n_{a}}-d_{a}\biggr]\cdot \frac{d}{d_{a}^2(d^2-1)} 
\ee 
for exponential encoding and Pauli re-uploading, respectively. Moreover, all distinguished Fourier coefficients have zero covariance and hence are uncorrelated. 
 \end{customthm}
\begin{proof}
Recall that
\be\label{Eq:FCSummandAgain2}
a_\omega = \sum_{u,v \text{ s.t. } \lambda_u- \lambda_v=\omega } a_{uv} = \sum_{i=1}^{g_\omega} a_{u_iv_i} \,,
\ee
where $\forall i \in \{0,1,\cdots,g_\omega \}: \lambda_{u_i}- \lambda_{v_i}=\omega$.

 First, we consider exponential encoding. By Eq.~\eqref{Eq:Omega}, the eigenvalues have exponentially increasing spacing, therefore $\lambda_u-\lambda_v=0$ if and only if $u=v$. By Eq.~\eqref{Eq:FCSummandAgain2}, we obtain 
 \begin{equation}
     a_0= \sum_{u=1}^{d_{a}} a_{uu}
 \end{equation}
 and by Proposition \ref{Prop:StatFS}, we obtain its variance
 \begin{align}  
 \label{AG:VarExp0}
     \Var(a_0) = &\sum_{u=1}^{d_{a}} \Var(a_{uu}) + \sum_{u,v = 1 , u\neq v}^{d_{a}} \Cov(a_{uu}, a_{vv}) \nonumber\\
     = &\sum_u  \frac{1}{(d+1)d_{a}^2} - \sum_{u\neq v} \frac{1}{(d^2-1)d_{a}^2} =\frac{1}{(d+1)d_{a}} - \frac{1}{(d^2-1)d_{a}^2}.(d_{a}^2-d_{a})  \nonumber \\ 
     = &  \frac{1}{(d+1)d_{a}} - \frac{d_{a}-1}{(d^2-1)d_{a}} \leq g_0^\text{Exp} \cdot \frac{d}{(d^2-1)d_{a}^2} \;.
 \end{align}
 We remark that, if $n_{h} = 0$ such that $d= d_{a}$, then for the exponential encoding we have $\mathbb{E}[|a_0|^2] = 0$. This means that for any reservoir and any observable, $a_0 = 0$. Thus the prediction after performing the linear regression does not contain an input-independent offset value corresponding to Fourier frequency $\omega=0$. In this case, it is necessary to include an additional parameter $\eta_0$ during training.

 For frequencies other than zero, since we know from Proposition \ref{Prop:StatFS} that the covariance of summands are zero, we have: 
 \be \label{AG:VarExpw}
\Var(a_\omega) = g_\omega \cdot \frac{d}{(d^2-1)d_{a}^2} 
 \ee
 where $g_\omega$ is the degeneracy given by Lemma~\ref{prop:3}. Eq. \eqref{AG:VarExpw} also holds for non-zero frequencies of Pauli encoding if we use Lemma~\ref{prop:2} to obtain the degeneracy. 

For QELMs with Pauli re-uploading, the degeneracy of frequency 0 is larger than exponential encoding. This means that there are $u,v$ such that $u \neq v$ and $\lambda_u-\lambda_v = 0$, and hence $a_{uv}$ is a summand contributing to $a_0$. We denote these sets of $u,v$ as $u_i, v_i$. Using Lemma~\ref{prop:2} and Eq. \eqref{Eq:FCSummandAgain2} we have: 
\begin{equation}\label{eq:pauli0}
a_0 =  \sum_{u=1}^{d_{a}} a_{uu} + \sum_{j=1}^{\binom{2n_{a}}{n_{a}}-d_{a}} a_{u_j,v_j}\,.
 \end{equation}
 where $u_j\neq v_j$ for all $j$. 
 Note that the first term is identical to the exponential encoding case and the second term consists of summands that according to Proposition \ref{Prop:StatFS} are linearly uncorrelated to themselves and the elements of the first term. Thus we obtain:  
 \begin{align}
     \Var(a_0) & = d_{a} \cdot \frac{1}{d_{a}^2(d+1)} - (d_{a}^2-d_{a})\cdot \frac{1}{d_{a}^2(d^2-1)} +\biggl[\binom{2n_{a}}{n_{a}}-d_{a}\biggr] \frac{d}{d_{a}^2(d^2-1)} \nonumber \\ 
     &\leq \binom{2n_{a}}{n_{a}} \cdot \frac{d}{d_{a}^2(d^2-1)} = g_0^\text{Pauli} \frac{d}{d_{a}^2(d^2-1)}
 \end{align}
\end{proof}
 As illustrated in Fig.~\ref{fig:VarEncoding}, the analytical results on the variances for both encodings agree with the simulations. Furthermore, Fig.~\ref{fig:MeanAbsCoef} shows the mean absolute value of the coefficients and standard deviations. Using Chebyshev's bound, We can show that the spectrum for Haar random reservoirs is rich (non-sparse) with high probability, which allows room for a potential quantum advantage, as discussed in Section~\ref{SecRoleR} and \ref{SecSurro}. Finally, the Fourier coefficients are all linearly uncorrelated since only summands of form $a_{ii}$, $a_{jj}$ have non zero covariances and they only contribute to coefficient $a_0$. Fig.~\ref{fig:Cov} shows the covariances between Fourier coefficients.
\begin{figure}
\centering
\begin{minipage}{0.49\textwidth}
  \centering
  \includegraphics[width=\textwidth]{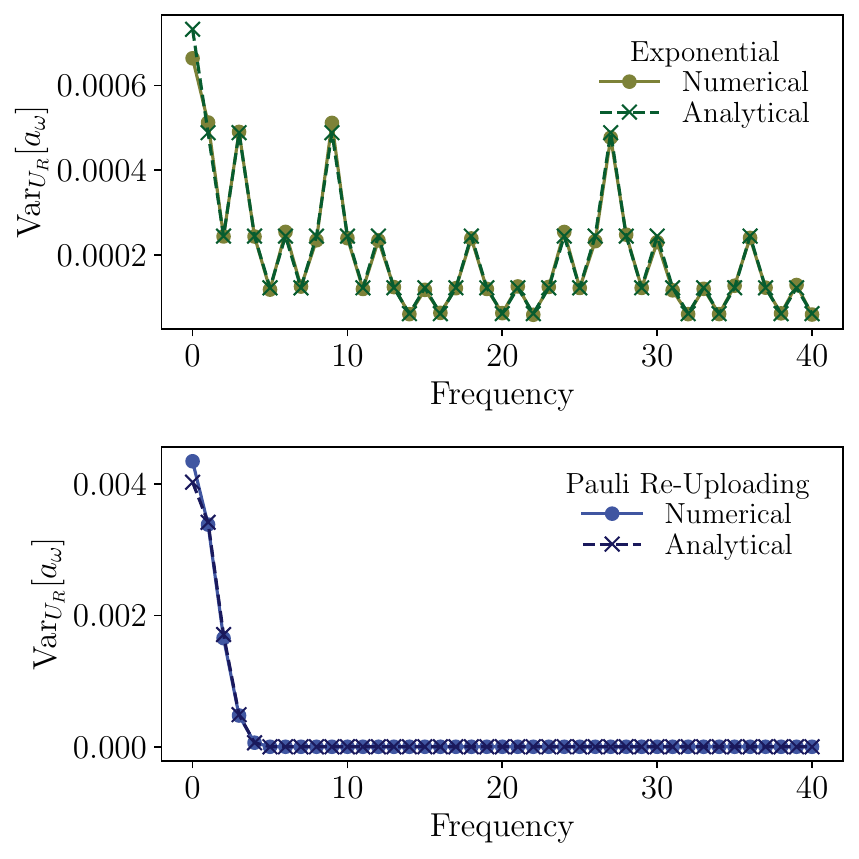}
  \label{fig:test1}
\end{minipage}
\hfill
\begin{minipage}{0.49\textwidth}
  \centering
\includegraphics[width=\textwidth]{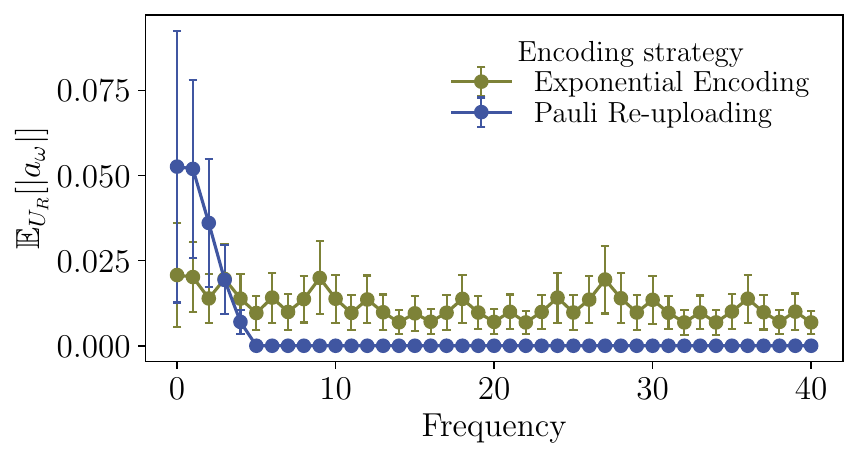}
   
  \label{fig:test2}
    \end{minipage}
    \caption{\textbf{Left: Variance of Fourier coefficients for different encodings.} We 
plot the variance of the Fourier coefficients $a_{uv}$ over the Haar random distribution 
with $n_{a} = 4$ encoding qubits and $n_{h}=2$ hidden qubits for an exponential encoding (up) and the Pauli re-uploading (down). The observable is a $Z$ measurement 
on the first encoding qubit. \textbf{Right: Mean absolute value of Fourier coefficients.} Here we plot the mean and standard deviation of the Fourier components over 1000 Haar random reservoir unitaries. The simulation is done with $n_{a} = 4$ encoding qubits and $n_{h} = 2$ hidden qubits and we consider measuring the $Z$ Pauli observable.}
\label{fig:VarEncoding}
\label{fig:MeanAbsCoef}
\end{figure}
\begin{figure}
    \centering   
    \includegraphics[width=0.45\textwidth]{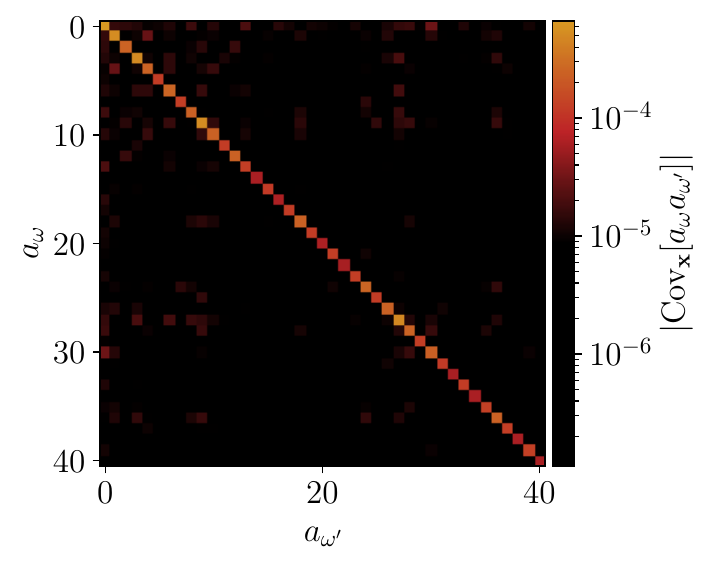}
    
    \caption{\textbf{Covariance matrix of Fourier coefficients.} The absolute value of the covariance matrix of the Fourier coefficients for the exponential encoding. The settings for numerics are the same as outlined in caption of Fig.~\ref{fig:MeanAbsCoef}. 
    }
    \label{fig:Cov}
\end{figure}

We remark that the recent work~\cite{mhiri2024gradients} discussed the statistics of Fourier coefficients of PQCs in detail and connect it to the expressivity of model prediction in generic scenarios.
\subsection{Classical surrogates of QELM models}
\subsubsection{Random Fourier Features method}
Random Fourier Features (RFF) method samples a number of frequencies from the frequencies existing in the model prediction $f_{\rm QELM}(x)$ according to distribution $p: \Omega \rightarrow \mathbb{R}$. Then, the function $f_{\rm RFF}(x)$ consisting of these frequencies are used to approximate the target $f_{\rm QELM}(x)$. 

In Ref.~\cite{sweke2023potential}, three conditions are stated for the RFF-dequantization of a PQC (i.e. such that the average error $ \mathbb{E}_{\vec{x}}\left[|f_{\rm QELM}(x)-f_{\rm RFF}(x)|^2\right]$ is low enough):
\begin{itemize}
    \item We need to be able to efficiently sample from the distribution $p(\omega)$ of frequencies with $\omega \in \Omega$.
    \item The distribution needs to be \textbf{concentrated} meaning that $p^{-1}_{max} \in \OC(\mathrm{poly}(n),\mathrm{poly}(d_x))$, where $d_x$ is the dimension of input data and $p_{\rm max}$ is the largest sampling probability i.e., $p_{\rm max}= \max_{\omega} p(\omega)$. In other words, the frequency spectrum of QELM's prediction should be concentrated.
    \item The frequency spectrum of the QELM output function is \textbf{well aligned} with the sampling distribution.
\end{itemize} 

Here we show that a distribution of Fourier frequencies $v : \Omega \rightarrow \mathbb{R}$ such that $v(\omega) \propto \mathbb{E}[|a_\omega|^2]$ i.e. it is proportional to the variances for exponential encoding, is anti-concentrated and hence not suitable for RFF-dequantization. In our case, we note that $\Ebb[|a_\omega|^2] = \Var[a_{\omega}]$ since $\Ebb[|a_{\omega}|]= 0$ from Proposition~\ref{Prop:StatFS}. Hence, the distribution can be expressed as $v(\omega) = c\cdot\Var[a_{\omega}]$ with $c\in\mathbb{R}$ being the proportional constant. Importantly, the largest redundancies happen at $\omega = 0$ which leads to $v_{\rm max} = v(0)$. Using Corollary \ref{col1} and Eq. \eqref{AG:VarExp0}, we obtain 
\begin{align}
    v_{\rm max} \leq \frac{c d}{(d^2-1)d_a}\;.
\end{align}
We observe that $v_{max}$ decreases exponentially in both the system size $n_{a}$ and the dimension of data $d_x$ meaning that the predition of a QELM with Haar reservoir has an anti-concentrated spectrum. More formally, we have the lower bound of $v_{\rm max}^{-1}$ as
\begin{align}
    v^{-1}_{\rm max}\geq\frac{(d^2-1)d_a}{c d} \;,
\end{align}
which contradicts the second requirement for the RFF surrogate.

Since the Haar distribution is a continuous distribution, there are instances of unitaries for which $|a_\omega|^2$ resembles the mean behaviour in Fig.~\ref{fig:VarEncoding} resulting in anti-concentrated spectrum. This means that at least for Haar random reservoirs with exponential encoding there are examples near the mean behaviour for which RFF dequantization is not efficient. There might be families of unitaries for which the three conditions above hold and thus an efficient RFF dequantization is possible. Moreover, for the instances where dequantization is impossible, exponential concentration might limit trainability of QELM. Overall, there seems to be a trade-off between the variance of output function and the concentration of Fourier coefficients.

\subsubsection{Trivial surrogate of QELM with Haar random reservoirs}
Importantly, the circuit being not efficiently surrogatable via RFF does not entail that it cannot be surrogated with other methods. One trivial surrogate can be constructed if the QELM suffers from exponential concentration (see Section~\ref{trainability_analysis} in the main text). In this scenario, the circuit becomes insensitive to inputs and outputs a constant independent of the input then a good surrogate of the circuit is simply that constant output. In particular, for the large number of qubits, we prove in this sub-section that guessing zero is already a good classical surrogate. 

Formally, denote $f_\text{Surr}(\bm{x})$ as the classical surrogate and $R := \mathbb{E}_{\bm{x}}[|f_{\bm{\eta}}(\bm{x}) -f_\text{Surr}(\bm{x})|]$ as the risk over the input space for a given reservoir $U_R$. In what follows, we show that by choosing the trivial ``guessing zero'' surrogate $f_\text{Surr}(\bm{x}) = 0$, the risk $R$ is an arbitrary small constant with the probability exponentially close to $1$ over the choice of reservoir $U_R$ for a large number of qubits $n$. That is, with $f_{\rm Surr}(\vec{x}) = 0$, we have
\begin{align}
    \Pr_{U_R}\left[ R \leq \epsilon \right] \geq 1 - \frac{\beta}{\epsilon^2} \;,
\end{align}
for some $\beta \in \OC(b^{-n}) $ with a constant $b > 1$ and some arbitrary small constant $\epsilon$. 

To prove this, we first consider the second moment of $R$ over the choice of the Haar random reservoir
\begin{align}
    \Ebb_{U_R}[R^2]  & =  \Ebb_{U_R} \left( \Ebb_{\vec{x}}\left[ \left| f_{\bm{\eta}}(\bm{x}) - f_\text{Surr}(\textbf{x}) \right| \right] \right)^2 \\
    & \leq \Ebb_{U_R} \Ebb_{\vec{x}} \left( f_{\bm{\eta}}(\bm{x}) - f_\text{Surr}(\textbf{x}) \right)^2 \\
    & =  \Ebb_{\vec{x}} \Ebb_{U_R} \left( f_{\bm{\eta}}(\bm{x}) - f_\text{Surr}(\textbf{x}) \right)^2 \\
    & =  \Ebb_{\vec{x}} \Ebb_{U_R} \left( f_{\bm{\eta}}(\bm{x}) \right)^2 \\
    &= \Ebb_{\bm{x},U_R}\left[\left(\sum_{i=1}^M \eta_i \expval{O_i}_{\bm{x}}\right)^2\right] \\
    & \leq \Ebb_{\bm{x},U_R}\left[\left(\sum_{i=1}^M \eta_i^2 \right) \left(\sum_{i=1}^M \expval{O_i}_{\bm{x}}^2 \right) \right] \\
    &\leq M B^2 \Ebb_{\bm{x},U_R}\left[\sum_{i=1}^M 
    \expval{O_i}_{\bm{x}}^2\right]\label{Eq:CauchyStd1}
\end{align}
where the first inequality is due to Jensen's inequality, in the fourth line we use explicitly that our surrogate is simply guessing zero $f_{\rm Surr}(\vec{x}) = 0$, the second inequality is by Cauchy-Schwarz inequality. Before proceeding we remark that since the trainable weights $\eta_i$'s depends on the choice of reservoir and the training inputs, they can not be simply taken out of the expectation directly. To circumvent this, in the third inequality we assume that all trainable weights $\eta_i$ are bounded by $B$ i.e., $|\eta_i | \leq B$ for $\forall i$. This assumption is a justified assumption since in almost all of regression tasks regularisation terms are included in to the loss function to ensure that $\eta_i$'s are not very large to avoid over-fitting.

Now, to further proceed, we will use explicitly the Haar-random reservoir assumption to do the average. That is, we have
\begin{align}
   \Ebb_{U_R} \expval{O_i}_{\bm{x}}^2 & = \Tr\biggr[\int_{\mathbb{U}(d)}d\mu(U_R) U(\bm{x})^{\dag\otimes2}U_R^{\dag\otimes2} O_i^{\otimes2} U_R^{\otimes2} U(\bm{x})^{\otimes2}\rho_0 ^{\otimes2}\biggl] \\
   & =   \frac{\Tr[O_i]^2\Tr[\rho_0]^2 + \Tr[O_i^2]\Tr[\rho_0^2]}{d^2-1} - \frac{\Tr[O_i^2]\Tr[\rho_0]^2 + \Tr[O_i]^2 \Tr[\rho_0^2]}{d(d^2-1)} \;, \label{eq:Haar-int}
\end{align}
where $d\mu(U_R)$ is a Haar measure over the unitary group $\UC(d)$ with $d = 2^n$ and in the second equality we perform Haar integration for the second moment over the group. We remark that the result of the integration does not depend on the input $x$ and the data-embedding $U(\vec{x})$. 

By substituting Eq.~\eqref{eq:Haar-int} into Eq.~\eqref{Eq:CauchyStd1}, we have that the risk averaged over the choice of Haar-random reservoir scales as
\begin{align}
    \Ebb_{U_R}[R^2]  & = MB^2 \sum_{i=1}^M \left[ \frac{\Tr[O_i]^2\Tr[\rho_0]^2 + \Tr[O_i^2]\Tr[\rho_0^2]}{d^2-1} - \frac{\Tr[O_i^2]\Tr[\rho_0]^2 + \Tr[O_i]^2 \Tr[\rho_0^2]}{d(d^2-1)}\right] \\
    & \in \OC\left(\frac{M^2 B^2 {\rm max}(\| O_i \|_{\infty})}{d} \right) \;,
\end{align}
where $\|\cdot\|_\infty$ is an infinity norm.

Next, since $R$ is non-negative, we realize that $ \Pr_{U_R}\left[ R \geq \epsilon \right] =  \Pr_{U_R}\left[ R^2 \geq \epsilon^2 \right]$. By invoking the Markov's inequality on $R^2$, we have
\begin{align}
    \Pr_{U_R}\left[ R \geq \epsilon \right] \leq \frac{\Ebb_{U_R} [R^2]}{\epsilon^2}  \;.
\end{align}

Lastly, by using the fact that $\Pr_{U_R}\left[ R \geq \epsilon \right] + \Pr_{U_R}\left[ R \leq \epsilon \right] = 1$, we can rearrange to obtain
\begin{align}
     \Pr_{U_R}\left[ R \leq \epsilon \right] \geq 1 - \frac{\beta}{\epsilon^2} \;,
\end{align}
with $\beta =\Ebb_{U_R} [R^2]$, which completes the proof of our claim.

\section{Practical consequences of exponential concentration on QELM}\label{appendix_prob_theory_refresher}
In this section, we investigate how exponential concentration of expectation values over input data $\xv$ and different instances of reservoir dynamics deteriorates the QELM performance. In particular, we analytically show that, when estimating these observables with a polynomial number of measurement shots, the trained QELM becomes insensitive to input data which leads to data-independent model predictions on unseen input data and in turn poor generalization.

\subsection{Refresher on some key concepts}
\subsubsection{Exponential concentration}
We begin with recalling the definition of probabilistic and deterministic exponential concentration discussed in the main text. 
We differentiate between \textit{probabilistic} and \textit{deterministic} exponential concentration:

\begin{definition}[Probabilistic exponential concentration]
Consider a quantity $Q(\vec{\theta})$ that depends on some variable $\vec{\theta}$ which can be estimated from an $n$-qubit quantum computer as an expectation value of some observable. $Q(\vec{\theta})$ is said to probabilistically exponentially concentrate around an input-independent value $\mu$ if
\be\label{eq:appx_def_exp_concentration}
\mathrm{Pr}_{\vec{\theta}}[|Q(\vec{\theta})-\mu|\geq\delta] \leq \dfrac{\beta}{\delta^2}\hspace{2pt},\;\; \beta \in \OC(1/b^{n}) \;,
\ee
for some $b > 1$. That is, the probability that $Q(\vec{\theta})$ deviates from $\mu$ by a small amount $\delta$ is exponentially small for all $\thv$.
\end{definition}

We remark that to satisfy Eq.~\eqref{eq:appx_def_exp_concentration} for probabilistically exponential concentration it is sufficient to show that the variance over the variables $\Var_{\thv}[Q(\vec{\theta})]$ is exponentially small with the number of qubits. This is due to Chebyshev's inequality which states that
\begin{align}
    {\rm Pr}_{\vec{\theta}} \left[\left|Q(\vec{\theta})-\Ebb_{\vec{\theta}}[Q(\vec{\theta})] \right|\geq\delta \right] \leq \dfrac{\Var_{\vec{\theta}}[Q(\vec{\theta})]}{\delta^2} \;,
\end{align}
where $\Ebb_{\thv}[Q(\vec{\theta})]$ is an average of $Q(\thv)$ over $\thv$. That is, we have $\mu = \Ebb_{\vec{\theta}}[Q(\vec{\theta})]$ and $\beta = \Var_{\thv}[Q(\vec{\theta})]$ in the presence of exponential concentration. We now recall the definition of deterministic exponential concentration:

\begin{definition}[Deterministic exponential concentration]
Consider a quantity $Q(\vec{\theta})$ that depends on some variable $\vec{\theta}$ which can be estimated from an $n$-qubit quantum computer as an expectation value of some observable. $Q(\vec{\theta})$ is said to deterministically exponentially concentrate around an input-independent value $\mu$ if its distance from it is bounded by an exponentially small value
\be
|Q(\vec{\theta})-\mu| \leq \beta \hspace{2pt},\;\; \beta \in \OC(1/b^n) \;,
\ee
for some $b >1$. That is, $Q(\vec{\theta})$ does not deviate from $\mu$ for more than $\beta$ for all $\theta$.
\end{definition}

In the context of QELM, we focus on two scenarios regarding the variables $\thv$. The first is when we have exponential concentration over the input data i.e., $\thv = \xv$ induced by the data encoding part while the other one is the concentration induced by reservoir dynamics i.e., $\thv = U_{R}$. 

\subsubsection{Estimating QELM predictions in practice}

We recall that a QELM model prediction with weights $\vec{\eta} = [\eta_1,...,\eta_M]^T$ is of the form
\begin{align} \label{eq:appx-model-prediction}
    f_{\vec{\eta}}(\xv) = \sum_{k=1}^M \eta_{k} \langle O_{k} \rangle_{\xv} \;,
\end{align}
with 
\begin{align}
    \langle O_{k} \rangle_{\xv,U_R} = \Tr[ O_{k} U_R (\rho(\vec{x})\otimes \ketbra{0}{0}) U^\dagger_R] \;,
\end{align}
where $O_{k}$ is an observable from a set $\{ O_{1} , ..., O_{M} \}$, $U_R$ is a reservoir unitary, $\rho(\vec{x})$ is an encoded quantum state associated with an input $\xv$ (or, simply an input-data state in the case of quantum data) and $\ketbra{0}{0}$ is the initial state in the hidden space.

In practice, the exact expectation values of these observables are inaccessible. Instead, their statistical estimates are obtained with measurement shots from quantum computers. More precisely, consider an expectation value of an observable $O = O_{k}$ with respect to a quantum state $\rho =  U_R (\rho(\vec{x})\otimes\ketbra{0}{0}) U^\dagger_R$. The observable can be decomposed as $O = \sum_{i} o_i |o_i \rangle \langle o_i |$ where $o_i$ and $|o_i\rangle$ are eigenvalues and associated eigenstates. After $N$ measurements, we can estimate the expectation value using the empirical mean of the outcome measurements as
\begin{align}
    \widehat{O} = \frac{1}{N} \sum_{m = 1}^N \lambda_m \;, 
\end{align}
where $\lambda_m$ is the $m^{\rm th}$ measurement outcome which can be treated as a random variable with a probability $\Tr[\rho |o_i\rangle\langle o_i |]$ to take a value $o_i$. Once we have gathered these statistical estimates for all the observables, the model prediction can be computed by classical post-processing of these estimates with the trainable weights.

\subsubsection{Hypothesis testing}
Here, we give some background overview of hypothesis testing which is an essential tool for showing the effect of exponential concentration on QELM. For more details about the topic, we refer the readers to Ref.~\cite{tsybakov2009introduction}.

\begin{lemma}\label{lemma:one-sample}
Consider two probability distributions $\PC$ and $\QC$ over some finite outcomes $\SC$. Suppose we are given a single sample $S$ drawn from either $\PC$ or $\QC$ (with an equal probability). We have the following hypotheses
\begin{itemize}
    \item Null hypothesis: $S$ is drawn from $\PC$
    \item Alternative hypothesis: $S$ is drawn from $\QC$
\end{itemize}
The success probability of making the right decision is
\begin{align}
    {\rm Pr}[``{\rm right \; decision}"] = \frac{1}{2} + \frac{\|\PC - \QC \|_1 }{4} \;
\end{align}
where $\norm{\PC-\QC}_1 = \sum_{s\in\SC}|p(s)-q(s)|$  is the distance between the two distributions expressed as a $1$-norm.
\end{lemma}
\begin{proof}
Define $\AC\subseteq \SC$ as a subset where $p(s) > q(s)$. The optimal decision we can make is to say that $s$ is drawn from $\PC$ if $s \in \AC$, otherwise it is drawn from $\QC$. Then, the probability of making the right decision is
    \begin{align}
        \Pr[``{\rm right \; decision}"]  &= \Pr[S \in \AC| S \sim \PC]\Pr[S\sim\PC]+\Pr[S \notin \AC| S \sim \QC]\Pr[S\sim\QC] \\
    &= \dfrac{1}{2}\biggr(\Pr[S \in \AC| S \sim \PC] + \Pr[S \notin \AC| S \sim \QC]\biggr) \\
        &= \dfrac{1}{2}\biggr(\sum_{s\in\AC} p(s) + \sum_{s\notin\AC}q(s)\biggr)\label{eq_A9}\,\;,
    \end{align}
where we use $\Pr[S\sim\PC] = \Pr[S\sim\QC] = \frac{1}{2}$ in the second line. 

Now consider the $1$-norm distance
    \begin{align}
        \norm{\PC-\QC}_1 = \sum_{s\in\SC}|p(s)-q(s)| = \sum_{s\in\AC}(p(s)-q(s))+\sum_{s\notin\AC}(q(s)-p(s))\,.    
    \end{align}
    Finally, let us expand the following quantity (using the fact that the probabilities sum to $1$)
    \begin{align}
        \dfrac{1}{2}(2+\norm{\PC-\QC}_1) &= \dfrac{1}{2}\biggr(\sum_{s\in\SC}p(s)+\sum_{s\in\SC}q(s)+ \sum_{s\in\AC}(p(s)-q(s))+\sum_{s\notin\AC}(q(s)-p(s))\biggr) \\
        &=\sum_{s\in\AC}p(s)+\sum_{s\notin\AC}q(s)\,,
    \end{align}
    which corresponds to the expression in Eq.~\eqref{eq_A9}, thus completing the proof.
\end{proof}

We now extend this hypothesis testing into a scenario where we obtain multiple samples (instead of a single sample). In this case, the following lemma holds
\begin{lemma}\label{lemma:many-sample}
Given a set of $N$ samples which is drawn from either $\PC^N$ or $\QC^N$ (with uniform probability), consider the two following hypotheses
\begin{itemize}
    \item Null hypothesis: samples are drawn from $\PC^N$
    \item Alternative hypothesis: samples are drawn from $\QC^N$
\end{itemize}
The success probability of making a right decision is upper bounded by
\begin{align}
{\rm Pr}[``{\rm right \; decision}"] \leq \frac{1}{2} + \frac{N\|\PC - \QC \|_1 }{4} \;.
\end{align}
\end{lemma}
\begin{proof} We can consider the whole sample set as an effective single sample drawn from either $\PC^N$ or $\QC^N$. By invoking Lemma~\ref{lemma:one-sample}, we have
\begin{align}
    {\rm Pr}[``{\rm right \; decision}"] & = \frac{1}{2} + \frac{\|\PC^N - \QC^N \|_1 }{4} \; \\
    & \leq  \frac{1}{2} + \frac{N \|\PC - \QC \|_1 }{4} \;,
\end{align}
where the last line is due to this useful identity $\|\PC^N - \QC^N \|_1 \leq N \| \PC - \QC \|_1$.
\end{proof}

\subsection{Statistical indistinguishability of QELM model predictions}

We are now ready to tackle a practical consequence of exponential concentration on QELM performance. In what follows, we provide analytical results indicating that 
the QELM model predictions become independent of unseen input data when data-dependent expectation values $\{\langle O_{k} \rangle_{\xv, U_R}\}_{k=1}^N$ are estimated with the polynomial number of measurement shots in the presence of exponential concentration. 

To begin with, we assume that the expectation values exponentially concentrate towards an input-independent point either over the input data $\xv$ or the reservoir dynamics $U_R$. That is, we have
\begin{align}\label{eq:appx-exp-con-qelm}
    \forall k\;, \;\; \Var_{\thv} \left[ \langle O_{k} \rangle_{\xv, U_R} \right] \leq \beta \;, \;\; \beta \in \OC(1/b^n)
\end{align}
where we have $\thv = \xv$ for exponential concentration over input data and $\thv = U_R$ for exponential concentration over a choice of reservoirs
\footnote{We only consider the effect of probabilistically exponential concentration as the same result can be directly applied to the deterministic case.}.
We note that sources that lead to exponential concentration are investigated in details in Appendix~\ref{appendix_exp_concentration}. We further consider the scenario where each observable in the set has half of the eigenvectors corresponding to $+1$ eigenvalue and the other half corresponding to $-1$ eigenvalue. In addition, the input-independent point $\mu$ is assumed to be zero that is
\begin{align}
    \forall k\;, \;\; \Ebb_{\thv} \left[ \langle O_{k} \rangle_{\xv, U_R} \right]  = 0 \;.
\end{align}
An example of this is a set of Pauli operators.

In the presence of exponential concentration, when estimating a given data-dependent expectation value from quantum computers, the probabilities of obtaining $\pm 1$ outcomes are both exponentially close to $1/2$. To see this, let us decompose the expectation value of $O_{k}$ as
\begin{align}
    \langle O_{k} \rangle_{\xv, U_R} = p_+^{(k)}(\xv, U_R) - p_-^{(k)}(\xv, U_R) \; ,
\end{align}
where $ p_\pm^{(k)}(\xv, U_R)$ are probabilities of $\pm 1$ outcomes. Together with the normalization condition $p_+^{(k)}(\xv, U_R) + p_-^{(k)}(\xv, U_R) = 1$, we have
\begin{align} \label{eq:appx-exp-con-qelm-avg-prob}
    \Ebb_{\thv} \left[ p_+^{(k)}(\xv, U_R)  \right] = \Ebb_{\thv} \left[ p_-^{(k)}(\xv, U_R)  \right] = \frac{1}{2} \;,
\end{align}
for the means of the probabilities, as well as
\begin{align}
    \Var_{\thv} \left[ \langle O_{k} \rangle_{\xv, U_R} \right] & = \Var_{\thv} \left[ p_+^{(k)}(\xv, U_R) \right] + \Var_{\thv} \left[ p_-^{(k)}(\xv, U_R) \right] - 2 \Cov_{\thv}\left[ p_+^{(k)}(\xv, U_R) , p_-^{(k)}(\xv, U_R) \right] \;, \\
    0 & = \Var_{\thv} \left[ p_+^{(k)}(\xv, U_R) \right] + \Var_{\thv} \left[ p_-^{(k)}(\xv, U_R) \right] + 2 \Cov_{\thv}\left[ p_+^{(k)}(\xv, U_R) , p_-^{(k)}(\xv, U_R) \right] \;,
\end{align}
which leads to
\begin{align}\label{eq:appx-exp-con-qelm-var-prob}
    \Var_{\thv} \left[ p_+^{(k)}(\xv, U_R) \right], \Var_{\thv} \left[ p_-^{(k)}(\xv, U_R) \right] \in \OC(1/b^n) \;,
\end{align}
for the variances of the probabilities. Therefore, we have the exponential concentration of outcome probabilities towards $1/2$. 

We are now ready to state the main formal result which addresses the practical consequence of exponential concentration. 

\begin{proposition}\label{prop:appx-stat-indis}
Assume exponential concentration of data-dependent expectation values as defined in Eq.~\eqref{eq:appx-exp-con-qelm}. For a given input data $\vec{x}$ and a choice of reservoirs $U_R$, define the probability distribution $\PC_{\vec{x}, U_R} = \{p_+^{(k)}(\xv, U_R), p_-^{(k)}(\xv, U_R) \}$. In addition, define a fixed data-independent distribution $\PC_0 = \{1/2, 1/2\}$. Given a set of $N$ samples/outcomes $\MC$ (with $N \in \OC(\poly(n))$) from either $\PC_0$ or $\PC_{\vec{x}, U_R}$ with an equal probability, consider the following two hypotheses
\begin{itemize}
    \item Null hypothesis: $\MC$ is drawn from $\PC^N_0$
    \item Alternative hypothesis: $\MC$ is drawn from $\PC^N_{\vec{x}, U_R}$
\end{itemize}
The success probability of making a right decision is exponentially close to a random guessing such that
\begin{align}
    {\rm Pr}[``{\rm right \; decision}"] & \leq \frac{1}{2} + \epsilon_n \;,\;\; \epsilon_n \in \OC(1/b'^n) \;,
\end{align}
for some $b'>1$. 
\end{proposition}

\begin{proof}
The success probability can be upper bounded as
\begin{align}
     {\rm Pr}[``{\rm right \; decision}"]  & = \int_{0}^1  {\rm Pr}\left(``{\rm right \; decision}" \Big|\, p_+^{(k)}(\xv, U_R) = p \right) {\rm Pr}\left(  p_+^{(k)}(\xv, U_R) = p \right) dp \\
     & \leq \int_{0}^1  \left[ \frac{1}{2} + \frac{N \| \PC_{\vec{x}, U_R} - \PC_0 \|_1}{4} \right] {\rm Pr}\left(  p_+^{(k)}(\xv, U_R) = p \right) dp \\
     & =  \int_{0}^1  \left[ \frac{1}{2} + \frac{N\left( | p - 1/2| + |(1-p)-1/2| \right)}{4}  \right] {\rm Pr}\left(  p_+^{(k)}(\xv, U_R) = p \right) dp \\
     & = \frac{1}{2} +  \frac{N}{2} \int_{0}^1  \left| p - \frac{1}{2}\right| {\rm Pr}\left( p_+^{(k)}(\xv, U_R) = p \right) dp \;,
\end{align}
where the first equality is due to Bayes' theorem which introduces the conditional probability of making the right decision given that $p_+^{(k)}(\xv, U_R) = p$ and then integrating all possible values of $p_+^{(k)}(\xv, U_R)$ to obtain the marginal probability and then the first inequality is by invoking Lemma~\ref{lemma:many-sample}. 

The integral can be interpreted as how far apart is $p$ from $1/2$ on average and intuitively this is exponentially small due to the exponential concentration. For convenience, denote $\sigma = \sqrt{ \Var_{\thv} \left[ \langle O_{k} \rangle_{\xv, U_R} \right]}$ and $\mu = 1/2$. We can further bound the integral as
\begin{align}
     \int_{0}^1  \left| p - \frac{1}{2}\right| {\rm Pr}\left(  p_+^{(k)}(\xv, U_R) = p \right) dp &=  \int_{0}^{1/2 - \sqrt{\sigma}}  \left| p - \frac{1}{2}\right| {\rm Pr}\left(  p_+^{(k)}(\xv, U_R) = p \right) dp \\
     &+ \int_{1/2 - \sqrt{\sigma}}^{1/2 + \sqrt{\sigma}}  \left| p - \frac{1}{2}\right| {\rm Pr}\left(  p_+^{(k)}(\xv, U_R) = p \right) dp   \\
     &+ \int_{1/2 + \sqrt{\sigma}}^1  \left| p - \frac{1}{2}\right| {\rm Pr}\left(  p_+^{(k)}(\xv, U_R) = p \right) dp \\
     \leq &  \int_{1/2 - \sqrt{\sigma}}^{1/2 + \sqrt{\sigma}}  \left| p - \frac{1}{2}\right| {\rm Pr}\left(  p_+^{(k)}(\xv, U_R) = p \right) dp + \sigma \\
     \leq &  \sqrt{\sigma} \int_{1/2 - \sqrt{\sigma}}^{1/2 + \sqrt{\sigma}}   {\rm Pr}\left(  p_+^{(k)}(\xv, U_R) = p \right) dp + \sigma \\
     \leq & \sqrt{\sigma} + \sigma \\
     \in & \OC\left( \frac{1}{b^{n/2}}\right) \;,
\end{align}
where the first inequality is by invoking Chebyshev's inequality i.e., ${\rm Pr}[|p - \mu| > \delta] \leq \sigma^2/\delta^2$ with $\delta = \sqrt{\sigma}$ and $\mu = 1/2$. The second inequality is by taking the maximum value of the integral and the third inequality is by extending the integration range back which leads to $\int_{0}^{1}   {\rm Pr}\left(  p_+^{(k)}(\xv, U_R) = p \right) dp =  1$. The last line follows from exponential concentration in Eq.~\eqref{eq:appx-exp-con-qelm-var-prob}.
\end{proof}

Proposition~\ref{prop:appx-stat-indis} implies that we cannot reliably distinguish between samples obtained from quantum computers and those obtained from a fixed data-independent distribution. In other words, estimates of expectation values have no information about input data. Consequently, QELM model predictions in Eq.~\eqref{eq:appx-model-prediction} which are obtained by classically post-processing these data-independent statistical estimates are clearly insensitive to input data. 

\section{Sources of exponential concentration for QELM}\label{appendix_exp_concentration}
\subsection{Expressibility-induced concentration}
First, recall the definition of the Schatten $p$-norm, whereby $\norm{X}_p = (\Tr[|X|^p])^{1/p}$ and $|X| = \sqrt{X\ad X}$. Particularly, $\norm{X}_1 = \Tr[|X|]$ and $\norm{X}_{\infty} = \sigma_{max}(X)$ where $\sigma_{max}(X)$ is the maximal eigenvalue of $X$.
Now, let us define the diamond norm of an arbitrary superoperator $S_A$ acting on a Hilbert space $\HC_A$
\be
    \norm{S_A}_{\diamond} := \sup_n \sup_{X_{AB}} \dfrac{\norm{(S_A \otimes \IC_{B}^{(n)})(\X_{AB})}_1}{\norm{X_{AB}}_1}\,,
\ee
where $X_{AB} \in \LC(\HC_A\otimes\HC_B)$ and $\IC_B^{(n)}$ is the identity acting on the $n$-dimensional Hilbert space $\HC_B$.

Now, we start by recalling and proving the encoding Haar-expressivity-induced theorem
\begin{theorem}[Encoding Haar-expressivity-induced concentration]
    Consider the expectation value of an arbitrary observable as defined in Eq.~\eqref{obs_def}. Then we have that
    \be
    \mathrm{Pr}_{\xv}[|\expval{O}_{\xv} - \Ebb_{\xv}[\expval{O}_{\xv}]| \geq \delta ] \leq \dfrac{G(\varepsilon_{\diamond}^{\xv})}{\delta^2}\,,
    \ee
    where 
    \be
    G(\varepsilon_{\diamond}) = \dfrac{\big(\Tr[\widetilde{O}_{\Lambda}]^2+\Tr[\widetilde{O}_{\Lambda}^2]\big)}{2^{n_{a}}(2^{n_{a}}+1)} + \varepsilon_{\diamond}^{\xv}\norm{\Lambda(O)}_{\infty}\,.
    \ee  
    where $\widetilde{O}_\Lambda = \Tr[(\IC_{A}\otimes\ketbra{0}{0})\Lambda(O)]$ with $\Lambda(\cdot)=U_{R}^\dag (\cdot)U_{R}$.
\end{theorem}

\begin{specialproof}
Let us note that $\Var_{\xv}[\expval{O}_{\xv}] \leq \Ebb_{\xv}[\expval{O}_{\xv}^2]$. Hence

\begin{align}
\Ebb_{\xv}[\expval{O}_{\xv}^2] &= \int_{\Ubb_{\xv}} \,dU(\xv) \Tr[\Lambda(O)\biggr(U(\xv)\rho_0U(\xv)\ad\otimes\ketbra{0}{0}\biggr)]^2 \\
&=\int_{\Ubb_{\xv}} \,dU(\xv) \Tr[\Lambda(O)^{\otimes 2}\biggr(U(\xv)^{\otimes 2}\rho_0^{\otimes 2}(U(\xv)\ad)^{\otimes 2}\otimes\ketbra{0}{0}^{\otimes 2}\biggr)] \\
&= \Tr[\Lambda(O)^{\otimes 2}\biggr(\biggr(\VC_{\text{Haar}}(\rho_0^{\otimes 2})-\AC_{\Ubb_s}(\rho_0^{\otimes 2})\biggr)\otimes\ketbra{0}{0}^{\otimes 2}\biggr)] \\
&\leq \Tr[\Lambda(O)^{\otimes 2}\biggr(\VC_{\text{Haar}}(\rho_0^{\otimes 2})\otimes\ketbra{0}{0}^{\otimes 2}\biggr)] + \norm{\Lambda(O)^{\otimes 2} \biggr(\AC_{\Ubb_s}(\rho^{\otimes 2})\otimes\ketbra{0}{0}\biggr)}_1 \\
&\leq \Tr[\Lambda(O)^{\otimes 2}\biggr(\VC_{\text{Haar}}(\rho_0^{\otimes 2})\otimes\ketbra{0}{0}^{\otimes 2}\biggr)] + \norm{\AC(\rho_0^{\otimes 2})}_1\norm{\Lambda(O)}_{\infty}^2 \\
&\leq \Tr[\Lambda(O)^{\otimes 2}\biggr(\VC_{\text{Haar}}(\rho_0^{\otimes 2})\otimes\ketbra{0}{0}^{\otimes 2}\biggr)] + \varepsilon_{\diamond}^{\xv}\norm{\Lambda(O)}_{\infty}^2\,,
\end{align}
where in the second equality we utilized the basic identity $\Tr[A]^2= \Tr[A^{\otimes 2}]$, in the third equality we inserted the definition of $\AC_{\Ubb_{\xv}}(\rho_0^{\otimes 2})$, in the first inequality we took the absolute value and applied the triangle inequality, in the second inequality we applied Hölder's inequality to the second term and the last inequality was achieved by noting that $\norm{\EC(X)}_1 \leq \norm{\EC}_{\diamond}\norm{X}_1$
Let us now compute explicitly the first term

\begin{align}
    \Tr[\Lambda(O)^{\otimes 2}\biggr(\VC_{\text{Haar}}(\rho_0^{\otimes 2})\otimes\ketbra{0}{0}^{\otimes 2}\biggr)] &= \int_{\UC(d)}\,d\mu(U) \Tr[\Lambda(O)^{\otimes 2} \biggr(U^{\otimes 2}\rho_0^{\otimes 2} (U\ad)^{\otimes 2}\otimes \ketbra{0}{0}^{\otimes 2}\biggr)] \\
    &= \Tr[\Lambda(O)^{\otimes 2} \biggr(\int_{\UC(d)}\,d\mu(U) U^{\otimes 2}\rho_0^{\otimes 2} (U\ad)^{\otimes 2}\otimes \ketbra{0}{0}^{\otimes 2}\biggr)] \\
    &= \dfrac{1}{2^{n_{a}}(2^{n_{a}}+1)}\Tr[\Lambda(O)^{\otimes 2}\biggr(\biggr(\IC^{\otimes 2}+\mathrm{SWAP}\biggr)\otimes\ketbra{0}{0}^{\otimes 2}\biggr)]\\
     &= \dfrac{\Tr[\widetilde{O}_{\Lambda}]^2 + \Tr[\widetilde{O}_{\Lambda}^2]}{2^{n_{a}}(2^{n_{a}}+1)}\,,
\end{align}
where the second equality is obtained by exchanging the trace and the Haar integral, and in the third equality we explicitly compute the second moment of the Haar integral. In the last equality we instead regroup the terms to recognise $\widetilde{O}_{\Lambda}$ and apply the identity $\Tr[\widetilde{O}_{\Lambda}^{\otimes 2}\mathrm{SWAP}] = \Tr[\widetilde{O}_{\Lambda}^2]$ where $\mathrm{SWAP}$ is the swap operator in the accessible space. Finally, we get

\be
\Var_{\xv}[\expval{O}_{\xv}] \leq \Ebb_{\xv}[\expval{O}_{\xv}^2] \leq \dfrac{\Tr[\widetilde{O}_{\Lambda}]^2 + \Tr[\widetilde{O}_{\Lambda}^2]}{2^{n_{a}}(2^{n_{a}}+1)} + \varepsilon_{\diamond}^{\xv}\lVert\Lambda(O)\rVert_{\infty}^2.
\ee
By applying Chebyschev's inequality, we complete the proof.
\end{specialproof}

\begin{corollary}\label{corollary1_expressibility}
    Suppose $O$ is a Pauli observable, i.e. $O \in \{X,Y,Z,I\}^{\otimes n}/\{I^{\otimes n}\}$ where $\{X,Y,Z\}$ are the single-qubit Pauli observables. Suppose furthermore that $\EC_{R}(\cdot) = U_{R}(\cdot)U_{R}^{\ad}$. Then we have
    
    \be
    G(\varepsilon_{\diamond})= \dfrac{1}{2^{n_{a}}+1} +\varepsilon_{\diamond}\,.
    \ee
\end{corollary}

Therefore, if the conditions of Corollary~\ref{corollary1_expressibility} are fulfilled and the encoding ansatz is expressible enough, so that $\varepsilon_{\diamond}^{\xv} \simeq 0$, then the expectation value will exponentially concentrate toward an input-independent value

\be
\mathrm{Pr}_{\xv}[|\expval{O}_{\xv} - \Ebb_{\xv}[\expval{O}_{\xv}]| \geq \delta ] \in \OC\biggr(\dfrac{1}{2^n}\biggr)\,.
\ee

\vspace{12pt}
We now turn to the reservoir unitary-induced exponential concentration. 
\begin{theorem}[Reservoir Haar-expressivity-induced concentration]
    Consider a reservoir evolution $U_R \in \Ubb_R$. Consider the expectation value of an arbitrary Hermitian observable as defined in Eq.~\eqref{obs_def}. Then we have that
    \be
    \mathrm{Pr}_{U_R}[|\expval{O}_{\xv} - \Ebb_{R}[\expval{O}_{\xv}]| \geq \delta ] \leq \dfrac{G(\varepsilon_{\diamond}^{R})}{\delta^2}\,,
    \ee
    where 
    \be
    G(\varepsilon_{\diamond}) = \dfrac{\big(\Tr[O]^2+\Tr[O^2]\big)}{2^n(2^{n}+1)} + \varepsilon_{\diamond}^{R}\norm{O}_{\infty}\,.
    \ee
\end{theorem}

\begin{specialproof}
    The proof is equivalent to the encoding Haar-expressivity-induced concentration.
\end{specialproof}

Notice that, if $U_{R}$ is drawn from a $2$-design, the average of the expectation value over $\Ubb_{R}$ is independent of the input $\xv$
\begin{align}
\Ebb_{U_R}[\expval{O}_{\xv}] &= \int_{\Ubb_{R}} \,dU_{R} \Tr[U_{R}\rho(\xv) U_{R}\ad O] = \int_{\UC(d)} \,d\mu(U) \Tr[U\rho(\xv) U\ad O] \\
&=\dfrac{\Tr[\rho(\xv)]\Tr[O]}{d} = \dfrac{\Tr[O]}{d} = \mu\,.
\end{align}
Crucially, Theorem~\ref{expressibility_thm2} tells us that if $\Ubb_{R}$ forms a $2$-design, the probability for the expectation value to differ from $\mu$ (which is independent of the input) by more than $\delta$ is exponentially small in the number of qubits. Hence, we will need exponentially many shots to recognise the observable from $\mu$.

\subsection{Entanglement-induced concentration}
\begin{theorem}[Entanglement-induced concentration]
    Suppose an observable that acts non-trivially on a subspace $\HC_{k}$ of the entire Hilbert space $\HC$, so that $O = O_{k} \otimes \IC_{\bar{k}}$. Then, the concentration of its expectation value around an input-independent will be bounded by
    \be
    \left| \expval{O}_{\xv} - \dfrac{\Tr[O]}{2^n}\right| \leq \norm{O_k}_\infty\sqrt{2\ln{2}}S\biggr(\rhot_k(\xv)\biggr\lVert \dfrac{\IC_k}{2^k}\biggr)^{1/2}\,,
    \ee
    where $S(\cdot\lVert\cdot)$ is the relative entropy and $\rhot_k(\xv) = \Tr_{\bar{k}}(\rhot(\xv))$ represents the final reduced state on subspace $\HC_k$.
\end{theorem}

\begin{specialproof}
Note that
\begin{align}
    \left| \expval{O}_{\xv} - \dfrac{\Tr[O]}{2^n}\right| &= \left| \Tr[O\biggr(\rhot(\xv) - \dfrac{\IC}{2^n}\biggr)]\right| \\
    &=\left|\Tr_k\biggr[O_k\biggr(\rhot_k(\xv)-\dfrac{\IC_k}{2^k}\biggr)\biggr]\right| \\
    &\leq \norm{O_k\biggr(\rhot_k(\xv)-\dfrac{\IC_k}{2^k}\biggr)}_1 \\
    &\leq \norm{O_k}_\infty \norm{\rhot_k(\xv)-\dfrac{\IC_k}{2^k}}_1 \\
    &\leq \norm{O_k}_\infty \sqrt{2\ln{2}} S\biggr(\rhot_k(\xv)\biggr\lVert\dfrac{\IC_k}{2^k}\biggr)^{1/2}\,,
\end{align}
where the first equality is given by using the definition of $\expval{O}_{\xv}$, the second equality by applying the definition of the observable as stated in Theorem \ref{entanglement_thm}, the first inequality is obtained by applying the triangle inequality, the second inequality by applying Hölder's inequality and for the final inequality we used Pinsker's inequality.
\end{specialproof}
\subsection{Global measurement-induced concentration}

\begin{theorem}[Global measurement-induced concentration]
Suppose an observable $O = \ketbra{m}{m}$, i.e. a projective measurement onto state $\ket{m}=\ket{m_1\dots m_n}$. Consider an initial separable state $\rho_0=\bigotimes_{k=1}^n \rho_0^{(k)}$. Suppose that the encoding unitary creates no entanglement, so that: $U(\xv)=\bigotimes_{k=1}^{n_{a}} U_k(x_k)$ where $x_k$ is an input component of $\xv$, and all are uniformly sampled from $[-\pi, \pi]$. Similarly, assume furthermore that the reservoir has the form $U_{R} =\bigotimes_{k=1}^n V_k$. Then we have
\be
\mathrm{Pr}_{\xv}[|\expval{O}_{\xv} - \Ebb_{\xv}[\expval{O}_{\xv}]\geq \delta] \leq \dfrac{\alpha\prod_{k=1}^{n_{a}} G_k(\varepsilon_{\Ubb_{x_k}})}{\delta^2}\,,
\ee
where $\varepsilon_{\Ubb_{x_k}}= \norm{\AC_{\Ubb_{x_k}}\biggr(\rho_0^{(k)^{\otimes 2}}\biggr)}_1$ is the Haar-expressivity measure of the local unitary $U_k(x_k)$ and $\alpha= \prod_{j=n_{a}+1}^{n}\left|\braket{0|V_j}{m_j}\right|^4$. The term $G_k(\varepsilon_{\Ubb_{x_k}})$ is given by

\be
G_k(\varepsilon_{\Ubb_{x_k}}) = \biggr(\frac{1}{3}+\varepsilon_{\Ubb_{x_k}}\biggr(\varepsilon_{\Ubb_{x_k}}+\sqrt{\frac{4}{3}}\biggr)\biggr)^{1/2}\,.
\ee
\end{theorem}

\begin{specialproof}
Let us start by noting that:

\begin{align}
\Ebb_{\xv}[\expval{O}_{\xv}^2] &= \int_{\Ubb_s} \,dU(\xv) \Tr[\rhot(\xv)\ketbra{m}{m}]^2 \\
&= \int_{\Ubb_{\xv}} \,dU(\xv) \Tr[(\rho(\xv)\otimes\ketbra{0}{0})\ketbra{m'}{m'}]^2 \\
&= \int_{\Ubb_{\xv}} \,dU(\xv) \Tr[(\rho(\xv)^{\otimes 2}\otimes\ketbra{0}{0}^{\otimes 2})\ketbra{m'}{m'}^{\otimes 2}],
\end{align}
where the first equality is obtained by using the definition of the observable, the second inequality is achieved by defining $\ket{m'} = \bigotimes_{k=1}^n\ket{m_k'} = V_{R}\ad\ket{m} =\bigotimes_{k=1}^{n} V_k\ad\ket{m_k}$ and the third equality is obtained by using the identity $\Tr[A]^2 = \Tr[A^{\otimes 2}]$. Then

\begin{align}
\int_{\Ubb_{\xv}} \,dU(\xv) \Tr[\rho(\xv)^{\otimes 2}\ketbra{m'}{m'}^{\otimes 2}] &= \int_{\Ubb_{\xv}} \,dU(\xv) \Tr[\biggr(U(\xv)^{\otimes 2}\rho_0^{\otimes 2} (U\ad(\xv))^{\otimes 2}\otimes\ketbra{0}{0}^{\otimes 2}\biggr)\ketbra{m'}{m'}^{\otimes 2}] \\
&= \prod_{k=1}^{n_{a}} \Tr[\int_{\Ubb_{x_k}}\,dU_k(x_k) U_k(x_k)^{\otimes 2} \big(\rho_0^{(k)}\big)^{\otimes 2}(U_k\ad(x_k))^{\otimes 2} \ketbra{m_k'}^{\otimes 2}] \prod_{j=n_{a}+1}^{n}\left|\braket{0}{m'_j}\right|^4
\\
&= \alpha\prod_{k=1}^{n_{a}} \Tr[\biggr(\VC_{\text{Haar}}\biggr(\big(\rho_0^{(k)}\big)^{\otimes 2}\biggr)- \AC_{\Ubb_{x_k}}\biggr(\big(\rho_0^{(k)}\big)^{\otimes 2}\biggr)\biggr) \ketbra{m_k'}^{\otimes 2}]
\\
&\leq \alpha\prod_{k=1}^{n_{a}}\norm{\biggr(\VC_{\text{Haar}}^{(k)}- \AC_{\Ubb_{x_k}}^{(k)}\biggr) \ketbra{m_k'}^{\otimes 2}}_1\label{global_meas_ineq_1}
\\ 
&\leq \alpha\prod_{k=1}^{n_{a}} \norm{\VC_{\text{Haar}}^{(k)}- \AC_{\Ubb_{x_k}}^{(k)}}_2 \norm{\ketbra{m_k'}^{\otimes 2}}_2 \label{global_meas_ineq_2} \\
&= \alpha\prod_{k=1}^{n_{a}} \biggr(\Tr[\biggr(\VC_{\text{Haar}}^{(k)}\biggr)^2] + \Tr[\AC_{\Ubb_{x_k}}^{(k)}\biggr(\AC_{\Ubb_{x_k}}^{(k)} - 2\VC_{\text{Haar}}^{(k)}\biggr)]\biggr)^{1/2} \\
&\leq \alpha\prod_{k=1}^{n_{a}} \biggr(\dfrac{1}{3} + \norm{\AC_{\Ubb_{x_k}}^{(k)}}_2\norm{\AC_{\Ubb_{x_k}}^{(k)} - 2\VC_{\text{Haar}}^{(k)}}_2\biggr)^{1/2} \label{global_meas_ineq_3} \\
&\leq \alpha\prod_{k=1}^{n_{a}} \biggr(\dfrac{1}{3}+\varepsilon_{\Ubb_{x_k}}\biggr(\varepsilon_{\Ubb_{x_k}}+\sqrt{\dfrac{4}{3}}\biggr)\biggr)^{1/2}\,, 
\label{global_meas_ineq_4}
\end{align}
where the first equality is straightforward by using the definition of the embedded state, the second equality is obtained by applying the separability of the encoding unitary and the initial state, as well as the projective $\ket{m'}$. The third equality is given by introducing the Haar-expressivity superoperator defined in Eq.~\eqref{expressibility_superop} for each single qubit and defining $\alpha=\prod_{j=n_{a}+1}^{n}|\braket{0}{m_j'}|^2$. In Eq. \eqref{global_meas_ineq_1} we define $\VC_{\text{Haar}}^{(k)} :=\VC_{\text{Haar}}\big(\big(\rho_0^{(k)}\big)^{\otimes 2}\big)$ and $\AC_{\Ubb_{x_k}}^{(k)} :=\AC_{\Ubb_{x_k}}\big(\big(\rho_0^{(k)}\big)^{\otimes 2}\big)$ to ease the notation, and we apply the triangle inequality. In Eq.~\eqref{global_meas_ineq_2} we instead apply Hölder's inequality. The last equality is obtained by expanding the first term and noting $\norm{\ketbra{m_k'}}_2 = 1$. Lastly, in Eq.~\eqref{global_meas_ineq_3} we explicitly compute the first term and apply Hölder's inequality to the second one, which is further bounded (in terms of $\varepsilon_{\Ubb_{x_k}}:=\norm{\AC_{\Ubb_{x_k}}^{(k)}}_2$) in Eq.~\eqref{global_meas_ineq_4} via the triangle inequality. Upon applying Chebyschev's inequality, the proof is complete.
\end{specialproof}

\begin{corollary}
    Suppose that each single-qubit unitary $U_k(x_k)$ is a 2-design, yielding $\varepsilon_{\Ubb_{x_k}}=0$. Then
    \be
        \mathrm{Pr}_{\xv}[|\expval{O}_{\xv} - \Ebb_{\xv}[\expval{O}_{\xv}]\geq \delta] \leq \dfrac{1}{\delta^2}\biggr(\dfrac{1}{3}\biggr)^{n_{a}/2}\,.
    \ee
\end{corollary}

\subsection{Noise-induced concentration}
Let us remind the statement of Theorem~\ref{noisy_encoding_thm}, which stems from considering a $L$-layered encoding subject to Pauli noise, as defined in Eq.~\eqref{noisy_encoding}
\begin{theorem}[Noise-induced concentration]\label{noise_QELM_thm}
    Consider the $L$-layered encoding as defined in  Eq.~\eqref{noisy_encoding} with $q < 1$. Then, the concentration around a fixed point of the expectation value of an observable as defined in Eq.~\eqref{obs_def} can be bounded as
    \be
    \biggr|\expval{O}_{\xv} - \dfrac{\Tr[\widetilde{O}_{\Lambda}]}{2^{n_{a}}}\biggr| \leq \norm{\Lambda(O)}_{\infty}\biggr(\dfrac{1}{b}q^{b(L+1)}S_2\biggr(\rho_0\biggr\lVert\frac{\IC}{2^{n_{a}}}\biggr)\biggr)^{1/2}\,,
    \ee
    where $\widetilde{O}_{\Lambda}= \Tr[\Lambda(O)(\IC_{A}\otimes\ketbra{0}{0})]$, $b = 1/(2\ln{2})$ and $S_2(\cdot\lVert\cdot)$ denotes the sandwiched $2$-Rényi relative entropy.
\end{theorem}

\begin{specialproof}
Let us bound the following quantity
\begin{align}
    \left| \expval{O}_{\xv} - \dfrac{\Tr[\widetilde{O}_{\Lambda}]}{2^{n_{a}}}\right| &= \left| \Tr[\Lambda(O)(\rho(\xv)\otimes\ketbra{0}{0})] - \dfrac{\Tr[\widetilde{O}_{\Lambda}]}{2^{n_{a}}}\right| = \left|\Tr[\Lambda(O)\biggr(\biggr(\rho(\xv)-\dfrac{\IC}{2^{n_{a}}}\biggr)\otimes\ketbra{0}{0}\biggr)]\right| \\
    &\leq\norm{\Lambda(O)\biggr(\biggr(\rho(\xv)-\dfrac{\IC}{2^{n_{a}}}\biggr)\otimes\ketbra{0}{0}\biggr)}_1 \\
    &\leq \norm{\Lambda(O)}_\infty \norm{\biggr(\rho(\xv)-\dfrac{\IC}{2^{n_{a}}}\biggr)\otimes\ketbra{0}{0}}_1 \\
    &\leq \norm{\Lambda(O)}_\infty \sqrt{2\ln{2}}S\biggr(\rho(\xv)\biggr\lVert \dfrac{\IC}{2^{n_{a}}}\biggr)^{1/2} \\
    &\leq \norm{\Lambda(O)}_\infty \sqrt{2\ln{2}}S_2\biggr(\rho(\xv)\biggr\lVert \dfrac{\IC}{2^{n_{a}}}\biggr)^{1/2} \\
    &\leq \norm{\Lambda(O)}_\infty \biggr(2\ln{2}q^{b(L+1)}S_2\biggr(\rho_0\biggr\lVert \dfrac{\IC}{2^{n_{a}}}\biggr)\biggr)^{1/2}\,,
\end{align}

where the first and second equality are obtained by using the definition of the expectation value of the observable and consequently grouping the two trace terms together. Then, in the first inequality we applied the triangle inequality, in the second one we used Hölder's inequality, in the third one Pinsker's inequality, in the fourth one we used the fact that $S(\cdot\lVert\cdot)\leq S_2(\cdot\lVert\cdot)$ and the final inequality is achieved by using the definition of $\rho(\xv)$ and applying a fundamental inequality for Pauli coefficients under noise channels (see Lemma 4 in Appendix E of Ref.~\cite{thanasilp2022exponential} for a detailed explanation)
\end{specialproof}

Theorem~\ref{noise_QELM_thm} can be easily extended to a $K$-layered reservoir noisy channel similar to Eq.~\eqref{noisy_encoding}), and with a characteristic noise parameter $p$. In this case we have

\be
    \biggr|\expval{O}_{\xv} - \dfrac{\Tr[O]}{2^n}\biggr| \leq \norm{O}_{\infty}\biggr(2\ln{2}q^{b(L+1)}p^{b(K+1)}S_2\biggr(\rho_0\biggr\lVert\frac{I}{2^n}\biggr)
    \biggr)^{1/2}\,.
\ee

\end{document}